\documentclass[11pt]{article}
\usepackage[margin=1in]{geometry}

\usepackage[utf8]{inputenc}
\usepackage[T1]{fontenc}

\usepackage{amsfonts}
\usepackage{amsmath}
\usepackage{amssymb}
\usepackage{amsthm}

\usepackage{graphicx}
\usepackage{placeins}
\usepackage{listings}
\usepackage{xcolor}
\usepackage{booktabs}
\usepackage{dsfont}
\usepackage{stmaryrd}
\SetSymbolFont{stmry}{bold}{U}{stmry}{m}{n}
\usepackage{enumitem}
\usepackage{pdfpages}
\usepackage{algorithm}
\usepackage{algorithmic}
\usepackage{bm}
\usepackage{bbm}
\usepackage{setspace}
\usepackage[hidelinks]{hyperref}
\usepackage{subcaption}
\usepackage{tikz}
\usetikzlibrary{positioning,arrows.meta}

\newtheorem{proposition}{Proposition}[section]
\newtheorem{definition}{Definition}[section]
\newtheorem{theorem}{Theorem}[section]

\newtheorem{remark}{Remark}[section]
\newtheorem{assumption}{Assumption}[section]

\newcommand{\cl}{\mathrm{cl}}
\newcommand{\op}{\mathrm{op}}


\title{Learning {Optimal Liquidation} with Closing Auctions}

\author{
  Julius \textsc{Graf}\footnote{\texttt{julius.graf@berkeley.edu}}
  \and
  Thibaut \textsc{Mastrolia}\footnote{\texttt{mastrolia@berkeley.edu}}
}

\date{
  UC Berkeley, Department of Industrial Engineering and Operations Research
  \\[1em]
  \today
}

\begin{document}
\maketitle

\begin{abstract} 
\begingroup
We study liquidation when continuous trading is followed by a closing auction. The trader first sells through a limit-order book, then submits signed auction schedules to adjust the remaining inventory. A projected clearing price signal and intermediate auction feedback connect the two phases. We compare deep Q-network (DQN) policies with projected deep deterministic policy gradient (DDPG), twin delayed deep deterministic policy gradient (TD3) and soft actor-critic (SAC) policies, using synthetic rough Heston prices and historical midprice paths within a simulated market. Policies are selected and evaluated by inventory-penalized implementation shortfall, separately from a weighted training objective. They achieve lower inventory-penalized shortfall than the stylized Avellaneda--Stoikov (AS) and time-weighted average price (TWAP) references, and matched synthetic comparisons show that auction access is useful for all four learners. We furthermore find that dense auction credit improves learning; the clearing forecast is informative, but its incremental decision value is learner-dependent.

\endgroup
\vspace{3mm}

\noindent{\bf Key words:} {optimal liquidation}, auction trading, reinforcement learning
\vspace{3mm}

\end{abstract}

\section{Introduction}

\subsection{Reinforcement Learning for Electronic Trading and Liquidation}

Market making is a cornerstone of modern electronic financial markets: market makers provide liquidity by continuously posting buy and sell quotes while managing inventory and adverse selection risk. Since the seminal work of Avellaneda and Stoikov \cite{avellaneda2008high} and its extension to explicit solutions in \cite{gueant2013dealing}, optimal market making has been studied extensively through stochastic control frameworks, leading to tractable strategies that balance expected profit against inventory risk under stylized assumptions on order flow dynamics and price evolution. See \cite{cartea2015algorithmic} for a review and \cite{baldacci2021algorithmic,bellia2025market} for recent advances. These models, however, typically rely on parametric assumptions that are difficult to validate empirically and may fail to adapt to the non-stationarity and strategic complexity of real-world markets.\vspace{0.5em}

The rapid growth of electronic trading and data availability combined with AI's increasing influence in the financial industry has motivated the use of reinforcement learning (RL) as a flexible, data-driven alternative to classical control methods. Reinforcement learning allows a market maker to learn optimal quoting policies directly from interaction with the market, without requiring full knowledge of the underlying dynamics. {Nevmyvaka et al. \cite{nevmyvaka2006reinforcement} combine Q-learning and dynamic programming for optimized trade execution using limit-order-book information. Reinforcement-learning studies include} {\cite{beysolow2019market,ning2021double,gavsperov2021market,gueant2019deep,jerome2023mbt,leal2022learning}, with regret analyses in \cite{cesa2024market,cao2024logarithmic}. Related predictive work uses deep neural networks to model future limit-order-book states \cite{sirignano2019deep}.} For broader treatments, see the reinforcement-learning survey \cite{hambly2023recent} and the volume on machine learning and data sciences for financial markets \cite{capponi2023machine}; \cite{cuchiero2024special} is an editorial introducing a special issue on machine learning in finance.\vspace{0.5em}

Most exchanges end the trading day with a closing auction. Inventory held into the closing auction can be liquidated at a single clearing price, but doing so exposes the market maker to auction-specific price impact, imbalance risk, and strategic interactions. Traditional market making models typically ignore the auction or treat the end-of-day liquidation in an ad hoc manner, while most reinforcement learning approaches focus exclusively on continuous trading on a central limit-order book (CLOB).\vspace{0.5em}

We consider a trader who begins with a positive inventory and submits sell limit orders during a continuous CLOB phase. At the auction opening, the control set is enlarged to signed auction schedules. The agent learns a feature-based policy without knowing the transition law of the exogenous order flow, while the exchange clearing mechanism itself is specified by the model.

\subsection{On the Importance of (Closing) Auctions}

Kyle \cite{kyle1985continuous} gives the first tractable equilibrium model of a continuous auction with asymmetric information; it explains how prices aggregate private information through order flow, and introduces a measure of market depth and price impact. This paper became the benchmark for analyzing liquidity, price discovery, and strategic trading in modern auction-based financial markets. While auction markets have been well-investigated in the economics community for discrete-time models (see, for example, \cite{milgrom1989auctions,vulcano2002optimal,milgrom2004putting}) and have attracted growing interest, especially since the work of Paul Milgrom, their modeling remains challenging both at high frequencies and in continuous time, and is identified as a major challenge in financial engineering in \cite[Section 5.1]{carmona2022influence}. \vspace{0.5em}

The trading session considered in this paper runs as follows. During continuous trading, the agent solves a one-sided end-of-day liquidation problem: she can only sell and her inventory remains nonnegative. During the closing auction, the admissible action set is enlarged to permit two-sided auction schedules, so that terminal inventory may lie on either side of zero. Deviations from zero terminal inventory are discouraged through a terminal inventory penalty. At the end of the auction phase, known as the clearing time, a clearing price is set by the exchange to maximize the volume traded. Closing auctions matter for liquidity and price discovery \cite{kandel2012effect,raillon2020growing}. We also refer to \cite{milgrom2019auction} for an overview of auction mechanisms.\vspace{0.5em}

Although their mechanism is distinct from the end-of-day inventory-management role studied here, frequent batch auctions have been proposed to mitigate the speed race associated with serial continuous trading, as explained in \cite{budish2014implementation}. A key challenge in the auction phase, known as the price discovery mechanism, is to set an efficient (clearing) price to trade the asset \cite{madhavan2000price,biais2002ipo,biais1999price}. Recent studies investigate the benefits of auctions for market quality, including reductions in the gap between the clearing price and the efficient price of a risky asset \cite{du2014welfare,paul2021optimal,derchu2024ahead,salek2024equity,jantschgi2024transaction}. Other work proposes incentives and optimal fee schemes to mitigate auction shortcomings \cite{mastrolia2026clearing,mastrolia2025optimal}.\vspace{0.5em}

Closing-auction volumes have risen substantially in recent years. Ignoring this stage is therefore problematic when modeling market making. For example, high-frequency market makers aim at ending each trading day with a flat position and are likely to use the auction to do so.

\begingroup
\subsection{Methodology, Contributions and Financial Insights}

The central question is whether a trader can use information about a future closing auction to balance continuous liquidation with terminal inventory management. We connect one-sided LOB execution to a closing call that accepts signed schedules. The trader must decide both how much inventory to carry into the call and how to adjust that position before clearing. A projected clearing signal links the phases, while intermediate feedback assigns learning credit before auction cash settles. Terminal positions may remain long or short; the inventory penalty discourages, rather than prohibits, residual exposure.\vspace{0.5em}

We formulate a single-agent reinforcement-learning problem \cite{sutton1998reinforcement}, connecting financial applications in allocation, execution and market making \cite{neuneier1997enhancing,ganesh2019reinforcement,kearns2013machine,ning2021double,baldacci2021algorithmic} to explicit closing-call control. A deep Q-network (DQN) \cite{watkins1992q,mnih2013playing,fan2020theoretical,mnih2015human,van2016deep} is the discrete-action baseline; deep deterministic policy gradient (DDPG), twin delayed deep deterministic policy gradient (TD3) and soft actor-critic (SAC) \cite{lillicrap2020continuous,fujimoto2018addressing,haarnoja2018soft} propose continuous controls that are projected onto the same executable grid.\vspace{0.5em}

The main (``headline'') policies train on a weighted preference objective, while policy comparisons instead minimize inventory-penalized implementation shortfall. Our references are a stylized Avellaneda--Stoikov (AS) policy adapted from \cite{avellaneda2008high,gueant2013dealing} and a time-weighted average price (TWAP) policy.\vspace{0.5em}

We address three empirical questions. First, do the learned policies use the closing call to improve liquidation and inventory management? Second, does forecast accuracy translate into incremental decision value? Third, does dense feedback help learn from delayed auction outcomes, separately from the additional preferences in the weighted objective?\vspace{0.5em}

The structure of this work is the following. Section \ref{sec:marketmodel} describes the general trading session structure investigated. The continuous trading phase and the auction phase are respectively described in Sections \ref{sec:lob} and \ref{sec:auction}. We then introduce the auction clearing mechanism in Section \ref{sec:clearing}, including existence of a clearing price for general supply–demand schedules (Theorem \ref{th:clearing}). Section \ref{sec:proj} proposes an algorithm to predict a hypothetical clearing price during the limit order book trading session and to anticipate the conversion of unmatched limit orders into trade blocks for the future closing auction. Section \ref{sec:MDP} defines admissible controls, the training objective and economic shortfall criteria. Section \ref{sec:learning} describes the learners, followed by the reference policies in Section \ref{sec:benchmark}, recalling the main results of \cite{avellaneda2008high, gueant2013dealing}, which we use as a stylized reference benchmark to compare the profit-and-loss (PnL) performance of a market maker. Section \ref{sec:numerics} presents the experimental protocol, simulator calibration and results.
\endgroup

\section{Market Model}\label{sec:marketmodel}

Throughout this work, we fix a probability space $(\Omega, \mathcal{F}, \mathbb{P})$, named the market, where $\Omega$ represents all the possible market configurations, $\mathcal F$ is a $\sigma$-algebra denoting the available information and $\mathbb P$ is the market probability. We consider a financial asset whose price evolves randomly during a trading session. We divide the trading period into two phases: a continuous phase on a limit order book with market makers and takers and an auction phase, seen as a usual closing auction. We fix two deterministic times $0< \tau^\op< \tau^\cl$ representing respectively the opening and the closing of the auction phase (in minutes) and denote by $h= \tau^\cl- \tau^\op$ the auction duration. We assume that $( \tau^\op,  \tau^\cl) \in \mathbb{N}^2$ and that $h$ is a positive integer, {with $\tau^\op\geqslant2$}. Therefore, the trading horizon is divided into a continuous phase $[ 0,  \tau^\op )$, in which the trader interacts with the central limit-order book, and a closing auction of fixed duration $h$ on $[ \tau^\op, \tau^\cl]$. {The agent acts on a realized decision grid $0=t_0<\cdots<t_n<\tau^\op=t_{n+1} <\cdots<t_m<\tau^\cl=t_{m+1}$. During the continuous phase, the decision times $t_0,\ldots,t_n$ are random times generated by the exogenous market-order processes, see Definition \ref{def:timing} below. In particular, $n$ is random. The decision clock is nevertheless control-independent: conditional on the same realization of the exogenous market processes, the sequence $(t_i)_{0\leqslant i\leqslant n}$ is the same for every admissible policy. During the auction phase, decisions are made at deterministic one-unit intervals such that $t_{n+1+\ell}=\tau^\op+\ell$ for $\ell \in \llbracket 0,h \rrbracket$, where $t_{n+1+h}=\tau^\cl$ is terminal. In particular, $m=n+h$.}\vspace{0.5em}
 
 We denote by $I_t$ the trader's inventory at time $t$. This inventory is positive for a long position and negative for a short one.\vspace{0.5em}
 
We denote by $\alpha > 0$ the tick size of the asset, fixed by the exchange. Three types of participant interact:
\begin{itemize}
\item a strategic trader, called the agent, who liquidates inventory through sell limit orders during continuous trading and submits signed schedules during the closing call,
\item exogenous market makers, posting limit orders and providing liquidity during the continuous trading session on both sides of the LOB and proposing limit prices during the auction session,
\item exogenous market takers. These participants submit aggressive market orders during both the CLOB and the auction phases to buy or sell the asset.
\end{itemize}
We focus on the agent's optimization throughout the day. The term ``exogenous'' indicates that we do not model the optimization problems of the other market makers and takers. All ``exogenous'' market participants are represented by a fixed stochastic data-generating mechanism (Section \ref{sub:generative}). This does not mean these participants have no strategic motives in an actual market, but that their individual optimization problems and endogenous responses to the agent are outside the model. We therefore consider a single-agent best-response problem against a specified background policy, not a market-equilibrium problem.

\begin{remark}
The two trading phases are non-overlapping by construction. We view this as a sequential, Xetra-like chronology rather than as a complete model of a US closing process in which auction order collection may overlap with continuous trading. In the historical-input experiments, US equity data determine only the midprice paths; order flow, auction updates and clearing remain synthetic.
\end{remark}

\subsection{Trading During the Continuous Phase \texorpdfstring{$[ 0,  \tau^\op )$}{}}\label{sec:lob}

During the continuous phase on the LOB, we assume that each market participant observes the midprice $S_t^\mathrm{mid} = \alpha k_t^\mathrm{mid}$ at any time $t< \tau^\op$, where $k^\mathrm{mid}_t\in \mathbb N$ is the midprice expressed in ticks. Exogenous market takers submit market orders on both sides of the market and consume liquidity. The number of market takers arriving on the side $\zeta\in\{+,-\}$ follows a counting process denoted by $N^\zeta$, where $\zeta = +$ denotes the ask side, and $\zeta = -$ the bid side. In other words, $N_t^\zeta$ market takers have arrived on the side $\zeta$ (e.g., if $\zeta=+$, then $N^+$ corresponds to buying market takers, consuming liquidity on the sell side). Each market taker with index $N_{t-}^\zeta<i\leqslant N_t^\zeta$ submits a market order of volume $\nu_t^{\zeta,i}$ at time $t$.\vspace{0.5em}

During the continuous phase, the agent acts as a one-sided liquidator and submits only sell limit orders. She submits limit orders characterized by a limit price denoted by $S_t^\bullet=\alpha k_t$ at time $t$ where $k_t$ denotes the number of ticks chosen at time $t$ to price the asset, and a proposed volume $v_t$. The agent therefore submits an order characterized by the pair $(k_t,v_t)$ at time $t$ on the limit order book. We assume that the agent is selling her inventory $I$ on the LOB during the continuous phase. This order is thus a limit order on $v_{t} \leqslant I_{t}$ shares at price $S_{t}^\bullet = \alpha k_t$.
\begin{remark}
    $\delta_{t} :=k_t-k_t^\mathrm{mid}$ represents the number of ticks between the sell limit order proposed by the agent and the midprice, seen as the ask-spread of the agent.
\end{remark}

The liquidity provided by exogenous market makers at price level $S_{t}^{\zeta,j} = \alpha k_{t}^{\zeta,j}$, where $k_{t}^{\zeta,j} = k_{t}^\mathrm{mid}+\zeta j$ for $j \in \mathbb{N}^\star$ and $\zeta \in \{+,-\}$, is given by the volume $V_{t}^{\zeta,j}$ at any time $t$. We assume that there exists $L_\infty\in\mathbb N^\star$ such that the depth of the order book on side $\zeta$ is given by
\[L_{t}^\zeta = \sup\{j\in \llbracket 1,L_\infty\rrbracket:V_{t}^{\zeta,j} >0\}\]
with the convention that $\sup\emptyset = 0$. To abstract from queue-position dynamics and latency competition, we adopt the following assumption.

\begin{assumption}
Conditional on the price level selected by the agent, her order has priority over exogenous orders resting at the same price level. We allow the agent to quote at the reference midprice, that is, $\delta_t=0$. Because the depth-ahead sum is empty for both $\delta_t=0$ and $\delta_t=1$, the modeled execution condition is identical at the two offsets. A midprice quote earns one tick less per executed share and therefore offers no additional execution advantage relative to quoting at $\delta_t=1$.
\end{assumption}

\begin{remark}
Conditional on a market taker arriving at time $t$, the agent's live sell order is at least partially executed if the following condition is satisfied:
\[
   \sum_{i=N_{t-}^++1}^{N_t^+}\nu_t^{+,i} > \sum_{\substack{j \in \mathbb{N}^* \\ j <\delta_t}} V_{t}^{+,j}.
\]
As soon as the buying volume of market takers reaches the agent's sell price level $k_t$, her order gets (at least partially) executed. The whole remark holds under the premise that the agent has a live sell order with positive remaining volume.
\end{remark}

For $i\in\llbracket0,n\rrbracket$, let $E_{t_i}$ denote the total number of shares executed over $(t_i,t_{i+1}]$ from the order submitted by the agent at $t_i$. This quantity is observed at $t_{i+1}$ and is given by
\[\forall i \in \llbracket 0,n-1 \rrbracket, \quad E_{t_i} = \max\left(0, \min\left(v_{t_i}, \sum_{j=N_{t_i}^++1}^{N_{t_{i+1}}^+}\nu_{t_{i+1}}^{+,j}- \sum_{\substack{j \in \mathbb{N}^* \\ j <\delta_{t_i}}} V_{t_i}^{+,j}\right)\right)\]
and
\[E_{t_n} = \max\left(0, \min\left(v_{t_n}, \sum_{j=N_{t_n}^++1}^{N_{{\tau^\op}_-}^+}\nu_{{\tau^\op}_-}^{+,j}- \sum_{\substack{j \in \mathbb{N}^* \\ j <\delta_{t_n}}}  V_{t_n}^{+,j}\right)\right).\]
We recall that $E_{t_i}$ is a random variable since $v_{t_i}$ and $V_{t_i}^{+,j}$ are random. The agent's inventory at the next decision time is $I_{t_{i+1}} = I_{t_i}-E_{t_i}$.\vspace{0.5em}

We assume that the agent will trade during this session until the deterministic physical time $t_n={\tau^\op}-1$ before the market switches to the closing auction phase, where $n$ is the final CLOB decision index; hence $t_0,\ldots,t_n$ comprise $n+1$ CLOB decisions.\vspace{0.5em}

Every order submitted by the agent at a CLOB decision time $t_i$ is good for the single decision interval $(t_i,t_{i+1}]$. Any unexecuted remainder expires without a cancellation fee at $t_{i+1}$. In particular, the order submitted at $t_n$ is active on $(t_n,\tau^\op]$ and any remaining quantity expires at $\tau^\op$.

\subsection{Trading During the Auction Phase \texorpdfstring{$[\tau^\op, \tau^\cl)$}{}} \label{sec:auction} At time $\tau^\op$, the system transitions to an auction, opened by the exchange. Similarly to \cite{derchu2024ahead,mastrolia2026clearing}, we assume that the agent submits bids at fixed deterministic times during the auction.  

\begin{assumption}
    The agent bids during the auction at discrete times $\tau^\op=t_{n+1}<\dots<t_m<\tau^\cl$.
\end{assumption}
{The remaining inventory $I_{\tau^\op}$ can be adjusted through the auction.} More precisely, for all $t \in \{t_{n+1},\dots,t_m\}$, the agent observes exogenous market and limit orders arriving in the auction. Market orders specify a volume to be bought or sold regardless of the clearing price, while limit orders are submitted during the auction through a supply function (functional volume to sell/buy above/below a certain price). The agent thus acts after observing the other participants' orders. After her final action at the terminal trading time $t_m<\tau^\cl$, the resulting book is frozen: no participant acts and no further event is processed on $(t_m,\tau^\cl]$; the clearing price, exchanged volume and terminal reward are then computed. Similarly to \cite{du2014welfare,paul2021optimal,mastrolia2025optimal} we assume that the agent submits a linear supply curve to the auction, as specified in the following assumption.

\begin{remark}
The ordering of moves immediately before clearing is a maintained stylization of the model, as the agent has a final reaction opportunity. {The numerical conclusions are conditional on this ordering; no separate cancellation effect is identified.}
\end{remark}

\begin{assumption}
    The agent has a linear supply curve $\Sigma_t:p\in \mathbb{R}_+ \longmapsto  K_t^a(p-S_t^a)$ for all $t \geqslant  \tau^\op$, where $K_t^a\geqslant 0$ and $S_t^a\in \alpha\mathbb N$.
\end{assumption}

At each $t_j\in\{t_{n+1},\dots,t_m\}$ the agent controls $K_{t_j}^a \geqslant  0$ and $S_{t_j}^a\in\alpha\mathbb N$ so that $\Sigma_{t_j}(p)$ represents net supply, that is, the number of shares the agent is willing to sell at price $p$. If $\Sigma_{t_j}(p)\leqslant 0$, the agent is willing to buy at price $p$ or below and, conversely, if $\Sigma_{t_j}(p)\geqslant  0$, the agent is willing to sell at price $p$ or above. Allowing both signs enlarges the auction control set beyond monotone liquidation. The objective nevertheless remains end-of-day inventory reduction, {encouraged through the terminal inventory penalty.}\vspace{0.5em}

{At auction decision time $t_j$, canceling past bids costs $d_{t_j}:=(j-n-1)d=(t_j-\tau^\op)d$, for $j\in\{n+1,\ldots,m\}$ and $d>0$.} The agent can cancel if and only if at least one prior strategic schedule is live, and, if she does so, removes every prior live schedule. Let $c_{t_j} \in \{0,1\}$ where $c_{t_j} = 1$ if and only if, at time $t_j$, the agent cancels all orders submitted prior to time $t_j$. In particular, at the opening of the auction, we set $c_{t_{n+1}} = 0$. We track these actions through a variable $\theta_{t_j}$, with $\theta_{t_n} = \mathbf{0}$ and $\theta_{t_{n+1}} = \mathbf 0$, where $\mathbf 0\in \mathbb{R}^{h}$ is the $(m-n)$-dimensional vector whose components are all zero. We recursively define
\[\theta_{t_j} =  \theta_{t_{j-1}}\vee  c_{t_{j-1}}\sum_{k=1}^{j-n-2}\mathbf{e}_{k}\mathbf{1}_{\{K_{t_{n+k}}^a>0\}}\]
for $j \in \llbracket n+2,m+1\rrbracket$, where $\mathbf{e}_{k} = (\delta_{i,k})_{1 \leqslant i \leqslant m-n}\in \mathbb{R}^{m-n}$. By definition, $\theta_{t_j}^{(s)}$ indicates whether the order submitted at time $t_{n+s}$ for $s \in \{1,\ldots,j-n-2\}$ has been canceled prior to time $t_j$. More precisely, $\theta_t^{(s)}=1$ if and only if the positive-slope schedule submitted in slot $s$ has already been canceled, and $\theta_t^{(s)}=0$ if the schedule is still live or if the slot contained no order. Therefore, each schedule can be canceled at most once.

\begin{remark}
    In this formulation, the cancellation fee $d_{t_j}c_{t_j}$ increases linearly with elapsed auction time $t_j-\tau^\op$ when cancellation occurs. We interpret this cost as a structural mechanism by the exchange to penalize ``illusory liquidity'', i.e., to discourage repeated provision and withdrawal of liquidity. The ``fictive'' auction reward introduced below is instead a noncash preference term in the training objective $J$.
\end{remark}

At each trading time $t\in \{t_{n+1},\dots, t_m\}$, the agent submits an order $(K_t^a, S_t^a, c_t)$ to the market. With the definition above, $c_t$ has no impact on the order $(K_t^a,S_t^a)$, only on orders sent at times strictly before $t$. The inventory of the agent remains frozen during the auction, i.e., $I_t = I_{ \tau^\op}$ for $\tau^\op\leqslant t<\tau^\cl$.\vspace{0.5em}

During the auction, we suppose in addition that both exogenous market makers and market takers are present in the auction. At each trading time $t$, the number of live exogenous market-maker schedules is denoted by $M_t$. Each bid sent by these actors is a limit order, each with volume $g_{i,t}(p)$, which is the supply schedule of the $i$-th bid offer present at time $t$ for $i \leqslant M_t$. Each exogenous $g_{i,t}$ is a signed net-supply schedule: positive values represent supply and negative values represent demand. \vspace{0.5em}

Market takers submit market orders in the auction. We retain the cumulative buy and sell arrival counters $N_t^+$ and $N_t^-$ from the continuous phase, without resetting them. For $\zeta\in\{+,-\}$, $N_{{\tau^\op}_-}^\zeta$ is the count immediately before the opening auction update. Auction orders present at time $t$ have indices $N_{{\tau^\op}_-}^\zeta<i\leqslant N_t^\zeta$, and $\nu_t^{\zeta,i}$ denotes their current volume. Each accepted auction arrival increments the corresponding counter by one and receives that index; a rejected proposal receives no index. At time $t_j$, an accepted cancellation of a previously active auction order is recorded as
\[
\nu_{t_j}^{\zeta,i}=0,\quad
N_{{\tau^\op}_-}^\zeta<i\leqslant N_{t_{j-1}}^\zeta,
\quad j\in\{n+2,\ldots,m\}.
\]
Earlier book snapshots remain unchanged. If an existing order is not canceled, its current volume is copied from $\nu_{t_{j-1}}^{\zeta,i}$ to $\nu_{t_j}^{\zeta,i}$. New accepted orders have indices $N_{t_{j-1}}^\zeta<i\leqslant N_{t_j}^\zeta$ for $j\geqslant n+2$, and $N_{{\tau^\op}_-}^\zeta<i\leqslant N_{\tau^\op}^\zeta$ at opening. A cancellation of a newly accepted order in the same exogenous update is likewise recorded only in the current snapshot, as permitted by Algorithm \ref{alg:generative_model}. Thus cumulative arrival indices are nondecreasing, whereas cancellations can reduce active order count and volume. All auction market-order sums below exclude the CLOB indices $i\leqslant N_{{\tau^\op}_-}^\zeta$.

\subsection{Clearing Price Rule and Estimation Along the Auction}\label{sec:clearing}

The final exogenous auction update is processed first, then the agent sends a final limit order and/or cancels past ones. After time $t_m$, the system moves into a final stage to set the clearing price of the auction. {At time $t=\tau^\cl$, the clearing rule selects a uniform tick price $S_{\tau^\cl}^\cl$ that maximizes matched volume, then allocates the long side pro rata. This pro-rata allocation rule is a modeling assumption. It preserves execution balance and treats participants proportionally without introducing order-level priority information absent from the model. First, we compute the continuous root $p_{m+1}^\star$ of aggregate net supply by solving the equation}
\begin{equation}
\label{eq:hyp_cl_auction}
    \sum_{i = 1}^{M_{t_m}}g_{i,t_m}(p) + \sum_{s=n+1}^{m} \left(1-\theta_{t_{m+1}}^{(s-n)}\right)K_{t_s}^a(p-S_{t_s}^a)- \sum_{\zeta \in \{+,-\}}\sum_{i=N_{{\tau^\op}_-}^\zeta+1}^{N_{t_m}^\zeta} \zeta \nu_{t_m}^{\zeta,i}= 0,\text{ for } p\in \mathbb{R}_+.
\end{equation}

During the auction, we assume that the agent computes the projected clearing price by solving the equation
\begin{equation}
\label{eq:hyp_cl_auction_estimation}
    \sum_{i = 1}^{M_{t_{j-1}}}g_{i,t_{j-1}}(p) + \sum_{s=n+1}^{j-1} \left(1-\theta_{t_j}^{(s-n)}\right)K_{t_s}^a(p-S_{t_s}^a)- \sum_{\zeta \in \{+,-\}}\sum_{i={N_{{\tau^\op}_-}^\zeta+1}}^{N_{t_{j-1}}^\zeta} \zeta \nu_{t_{j-1}}^{\zeta,i}= 0,\text{ for } p\in \mathbb{R}_+
\end{equation}
which is the clearing price equation were the auction to close at time $t_j \in \{t_{n+2},\ldots, t_{m+1}\}$, and then applying the volume-maximizing tick projection to obtain an element of $\alpha \mathbb{N}$. At $t_j=t_{m+1}=\tau^\cl$, the estimate is obtained by solving \eqref{eq:hyp_cl_auction} and applying the same tick projection as the exchange's clearing rule.

\begin{remark}
For $j\geqslant  n+2$, the label $H_{t_j}^\cl$ denotes the indicative price available to the agent at decision time $t_j$ obtained by solving \eqref{eq:hyp_cl_auction_estimation} at time $t_j$ and applying the volume-maximizing tick projection. It is computed from the exogenous auction book recorded at $t_{j-1}$ and from all still-active agent schedules submitted up to and including $t_{j-1}$. Consequently, the time index of $H_{t_j}^\cl$ records when the estimate becomes available instead of the date of the underlying book snapshot.
\end{remark}

We now provide sufficient conditions ensuring existence of a solution to \eqref{eq:hyp_cl_auction_estimation}. First, Equation \eqref{eq:hyp_cl_auction_estimation} can be written as $\Phi_j(p)=0$ where
\begin{align*}
    \Phi_j(p) = \sum_{i=1}^{M_{t_{j-1}}} g_{i,t_{j-1}}(p) + A_jp-B_j-C_j
\end{align*}
with
\[A_j = \sum_{s=n+1}^{j-1} \left(1-\theta_{t_j}^{(s-n)}\right)K_{t_s}^a, \quad B_j = \sum_{s=n+1}^{j-1} \left(1-\theta_{t_j}^{(s-n)}\right)K_{t_s}^aS_{t_s}^a, \quad C_j = \sum_{\zeta \in \{+,-\}}\sum_{i={N_{{\tau^\op}_-}^\zeta+1}}^{N_{t_{j-1}}^\zeta} \zeta \nu_{t_{j-1}}^{\zeta,i}.\]
Furthermore, let $A_{i} = B_{i} = 0$ for $i \leqslant n+1$.

\begin{theorem}[Existence of a unique continuous clearing root]\label{th:clearing}
Let $t_j \in \{t_{n+2},\ldots,t_{m+1}\}$. Assume that $\lim_{p\to +\infty} g_{i,t_{j-1}}(p)=+\infty$
and $p\longmapsto g_{i,t_{j-1}}(p)$ is continuous and strictly increasing for any $i\leqslant M_{t_{j-1}}$. Assume moreover $\Phi_j(0)\leqslant 0$, and the condition ``$M_{t_{j-1}}\geqslant 1$ or $A_j > 0$'' holds. Then, the estimated clearing price equation \eqref{eq:hyp_cl_auction_estimation} admits a unique solution in $\mathbb R_+$. 
\end{theorem}

\begin{proof}
Clearly, $\Phi_j$ is continuous, $\lim_{p\to +\infty} \Phi_j(p) = +\infty$ and for $q>p$,
\[\Phi_j(q)-\Phi_j(p) = \sum_{i=1}^{M_{t_{j-1}}} \left[g_{i,t_{j-1}}(q)-g_{i,t_{j-1}}(p)\right] + A_j(q-p), \]
where $A_j \geqslant 0$. If $M_{t_{j-1}}\geqslant 1$, the first term is strictly positive, and if $M_{t_{j-1}}=0$, then $A_j>0$ and the second term is strictly positive. Hence $\Phi_j$ is strictly increasing. The intermediate value theorem gives existence of a zero, and strict monotonicity gives uniqueness. Since $\Phi_j(0)\leqslant 0$, the zero is nonnegative.
\end{proof}

Theorem \ref{th:clearing} covers nonlinear background schedules. In practice, market participants submit collections of limit orders which the exchange aggregates to reconstruct the supply/demand functions. We refer to \cite{mastrolia2025optimal} for an example. These supply/demand functions need not be linear. For linear supply and demand curves, the continuous clearing root has the form used in \cite{paul2021optimal} or \cite{mastrolia2025optimal}, as stated in the following proposition. 
\begin{proposition}
\label{prop:linear}
   Fix $j\in\{n+2,\ldots,m+1\}$. Assume that $g_{i,t}(p) = K_t^i(p-S_t^i)$ with $K_t^i \geqslant 0$, $\sum_{i=1}^{M_{t_{j-1}}} K_{t_{j-1}}^i +A_j >0$ and $\sum_{i=1}^{M_{t_{j-1}}} K_{t_{j-1}}^iS_{t_{j-1}}^i +B_j +C_j \geqslant 0$. Then the unique continuous clearing root solving \eqref{eq:hyp_cl_auction_estimation} is  
   \[p = \frac{\sum_{i=1}^{M_{t_{j-1}}} K_{t_{j-1}}^iS_{t_{j-1}}^i +B_j +C_j}{\sum_{i=1}^{M_{t_{j-1}}} K_{t_{j-1}}^i +A_j}\geqslant 0.\]
\end{proposition}

For the linear specification used below, we consider background schedules satisfying $K_t^i>0$ and impose the conditions of Proposition \ref{prop:linear} at every auction evaluation:
\[D_j = \sum_{i=1}^{M_{t_{j-1}}} K_{t_{j-1}}^i +A_j > 0 \quad \text{and} \quad \sum_{i=1}^{M_{t_{j-1}}} K_{t_{j-1}}^iS_{t_{j-1}}^i +B_j +C_j \geqslant 0.\]

\begingroup
For each continuous root $p_j^\star$ in Equation \eqref{eq:hyp_cl_auction_estimation}, define its two bracketing ticks $p_j^-:=\alpha\lfloor p_j^\star/\alpha\rfloor$ and $p_j^+:=\alpha\lceil p_j^\star/\alpha\rceil$, and let $\mathcal P_j:=\{p_j^-,p_j^+\}$. Let $Q_j^S(p)$ and $Q_j^D(p)$ denote total supply and demand in the same book used to compute $p_j^\star$, with the positive/negative-part definitions given below for the terminal book. Include market orders and aggregate all live strategic schedules into one net schedule before splitting its sign. Write $M_j(p):=\min\{Q_j^S(p),Q_j^D(p)\}$. The volume-maximizing projection, used for both lagged auction indications and terminal clearing, is
\[
p_{\alpha,j}:=\underset{p\in\mathcal P_j}{\operatorname{arg\,max}}\,
\bigl(M_j(p),-\lvert p-p_j^\star\rvert,p\bigr),
\]
where the tuple is ordered lexicographically: maximize matched volume first, then choose the tick nearest the root, then the higher tick. An on-grid root has one candidate. Monotone supply and demand ensure that a global tick-volume maximizer is among the two bracketing ticks. For linear books, nearest-root rounding minimizes absolute imbalance, whereas the volume-maximizing projection used here need not do so. At the terminal evaluation, write $p_\alpha=p_{\alpha,m+1}=S_{\tau^\cl}^\cl$. The selected tick is at most one tick from the root, so for the linear aggregate
\[
\lvert Q_j^S(p_{\alpha,j})-Q_j^D(p_{\alpha,j})\rvert
=D_j\lvert p_{\alpha,j}-p_j^\star\rvert\leqslant D_j\alpha.
\]

\endgroup\vspace{0.5em}

{The agent's requested signed quantity at the selected clearing tick is}
\[\tilde Z_{ \tau^\cl} = \sum_{s= n+1}^{m}\left(1-\theta_{t_{m+1}}^{(s-n)}\right)
K_{t_s}^a\left(S_{\tau^\cl}^\cl - S_{t_s}^a\right).\]
Let
\[\forall p \in \mathbb{R}_+,\quad \Sigma_{\tau^\cl}^a(p) = \sum_{s= n+1}^{m}\left(1-\theta_{t_{m+1}}^{(s-n)}\right)
K_{t_s}^a\left(p - S_{t_s}^a\right).\]
At the {selected terminal tick} $p_\alpha$, total
supply and demand are
\[
Q^S(p_\alpha) = \sum_{i=1}^{M_{t_m}}g_{i,t_m}(p_\alpha)_+ + \Sigma_{\tau^\cl}^a(p_\alpha)_+ + \sum_{i={N_{{\tau^\op}_-}^-+1}}^{N_{t_m}^-}\nu_{t_m}^{-,i},\]
and
\[
Q^D(p_\alpha) = \sum_{i=1}^{M_{t_m}}(-g_{i,t_m}(p_\alpha))_+ + (-\Sigma_{\tau^\cl}^a(p_\alpha))_+ + \sum_{i={N_{{\tau^\op}_-}^++1}}^{N_{t_m}^+}\nu_{t_m}^{+,i}. \]
Let
\[ \rho_S = \left\{\begin{array}{ll} 1 & \text{if }Q^S=0\text{ or }Q^S\leqslant Q^D\\ Q^D/Q^S & \text{if }Q^S>Q^D \end{array}\right. \quad \text{and} \quad \rho_D = \left\{ \begin{array}{ll} 1 & \text{if }Q^D=0\text{ or }Q^D\leqslant Q^S\\ Q^S/Q^D & \text{if }Q^D>Q^S \end{array}\right..\]
The signed net auction fill of the agent at the clearing time is $Z_{\tau^\cl} = \rho_S(\tilde Z_{\tau^\cl})_+ - \rho_D(-\tilde Z_{\tau^\cl})_+.$ The agent's active schedules are aggregated into a single net schedule before this allocation, so no self-trade is counted.

\subsection{Projected Hypothetical Clearing Price During the Continuous Session}\label{sec:proj}

{During the continuous phase $[0,\tau^\op)$, the simulator uses the standing exogenous CLOB to construct a raw projected clearing estimator $H_t^\circ$. It treats the outstanding exogenous orders as inputs to a fictitious auction and calibrates their linear-schedule representation through the volume moments in Algorithm \ref{alg:hyp_clearing_price}, following the construction motivated by \cite{paul2021optimal}. A separate training-only reliability fit, defined immediately after the algorithm, converts this raw estimator into the CLOB forecast $H_t^\cl$ supplied to the policy.}\vspace{0.5em}

Since the previous strategic order has expired by $t_i$, $O_i$ denotes the standing exogenous CLOB multiset observed before the action at $t_i$.

\begin{algorithm}[htbp!]
\caption{{Computation of the raw estimator $H_{t_i}^\circ$ for $i \in \llbracket 0, n \rrbracket$}}
\label{alg:hyp_clearing_price}
\begin{algorithmic}[1]
\REQUIRE Tick size $\alpha>0$, smoothing parameter $\eta_H\in(0,1]$, initial value $H_0$
\STATE {Initialize $H_{t_0}^\circ \gets H_0$}
\FOR{$i = 1, \dots, n$}
\STATE Record standing orders before time $t_i$ by multiset $O_i \in \mathcal{M}_f (\alpha \mathbb{N}\times \mathbb{R}_+)$\COMMENT{$\mathcal{M}_f(\cdot)$ denotes the finite multisets.}
\IF{$O_i \neq \emptyset$}
\STATE $\mathcal{K}_i \gets \operatorname{proj}_1(O_i)/\alpha$ \COMMENT{Standing price levels}
\STATE Register $Q_i(k)= \sum_{(p,v)\in O_i} v\delta_{p,\alpha k}$ for $k \in \mathcal{K}_i$ and $Q_i(k)=0$ if $k\notin \mathcal{K}_i$
\FOR{$k\in \mathcal{K}_i$}
      \STATE $\hat e_i^k \gets \frac{1}{i} \sum_{s = 1}^i Q_s(k)$ \COMMENT{Average volume available at level $k$, with $\hat e_i^k>0$ for $k\in \mathcal{K}_i$}
      \STATE $\hat \varsigma_i^k \gets \frac{1}{i} \sum_{s = 1}^i Q_s(k)^2$ \COMMENT{Average squared volume available at level $k$}
      \STATE $\hat K_i^k \gets  \max(0,(2\hat e_i^k-\hat \varsigma_i^k/\hat e_i^k)\alpha^{-1})$ \COMMENT{Calibrated slope at level $k$}
\ENDFOR
\IF{$\sum_{k \in \mathcal{K}_i} \hat{K}_i^k > 0$}
\STATE Solve $\sum_{k \in \mathcal{K}_i}\hat{K}_i^k(\alpha k-p) = 0$ for $p$ and denote the solution $\tilde S_{t_i}$ \COMMENT{Clearing price rule}
\STATE {$H_{t_i}^\circ \gets H_{t_{i-1}}^\circ + \eta_H(\tilde S_{t_i} - H_{t_{i-1}}^\circ)$} \COMMENT{Smoothed update rule}
\ELSE
\STATE {$H_{t_i}^\circ \gets H_{t_{i-1}}^\circ$}
\ENDIF
\ELSE
\STATE {$Q_i(k)\gets 0$ for every price level $k$}
\STATE {$H_{t_i}^\circ \gets H_{t_{i-1}}^\circ$}
\ENDIF
\ENDFOR
\end{algorithmic}
\end{algorithm}

\begingroup
Before using the raw estimate as a forecast, we shrink its displacement from the current midprice:
\[
H_t^\cl=S_t^{\mathrm{mid}}+w_{b(t)}\bigl(H_t^\circ-S_t^{\mathrm{mid}}\bigr).
\]
Here $b(t)\in\llbracket 0,l \rrbracket$ identifies the time bin containing $t$. We consider $l=3$ from now on. For example, when working on a time horizon of 120 minutes, the four equal quarters before opening are $[0,30)$, $[30,60)$, $[60,90)$ and $[90,120)$ minutes in the numerical setting. This correction neither feeds back into Algorithm \ref{alg:hyp_clearing_price} nor rounds the CLOB forecast to a tick. Initially, the raw estimate and midprice both equal a common initial value (100 in our simulations). We fit the weights separately on 256 no-order training paths for each setting. Within each bin, we regress the eventual {clearing price}'s displacement from the current midprice on the corresponding raw-estimate displacement. The regression has no intercept and constrains the coefficient to $[0,1]$. Historical calibration draws only from the pooled training sessions, and all order flow and clearing on these paths are simulated. The fitted weights for the four successive quarters are $(0.075,0.112,0.225,0.340)$ on synthetic paths and $(0.192,0.312,0.384,0.523)$ on historical paths. Thus the fit puts more weight on the raw estimate as closing approaches. The weights are then frozen across learners and treatments. This forecast calibration is separate from the moment-based slope calculation and the residual-book conversion below.\vspace{0.5em}

At auction opening, the environment inherits the last calibrated pre-action CLOB forecast:
\[
H_{\tau^\op}^\cl=H_{t_n}^\cl
=S_{t_n}^{\mathrm{mid}}+w_{b(t_n)}\bigl(H_{t_n}^\circ-S_{t_n}^{\mathrm{mid}}\bigr).
\]
For $j=n+2,\ldots,m$, let $p_j^\star$ be the continuous root of \eqref{eq:hyp_cl_auction_estimation}, computed from the complete post-action book at $t_{j-1}$, including all strategic schedules still live after that action. The next decision's indicative price is $H_{t_j}^\cl=p_{\alpha,j}$. The indicative price uses the same volume-maximizing projection as terminal settlement, applied to this lagged book. It is an indication of hypothetical clearing, with no execution before $\tau^\cl$. Current exogenous proposals at $t_j$ are accepted before the agent observes the current features and acts, but do not update this lagged indicative price. No CLOB reliability weight is applied during the auction. The final post-action book at $t_m$ is frozen; at $t_{m+1}=\tau^\cl$, $H_{\tau^\cl}^\cl=S_{\tau^\cl}^\cl$ is its volume-maximizing terminal clearing price defined in Section \ref{sec:clearing}.\vspace{0.5em}
\endgroup

Let $O_{{\tau^\op}_-}^{\star}$ denote the multiset of exogenous CLOB orders remaining after all executions on $(t_n,\tau^\op]$. The unexecuted remainder of the agent's order submitted at $t_n$ expires at $\tau^\op$ and is excluded from this multiset. To map the residual CLOB book into the linear-schedule representation
used in the auction, we repeat the final calibration step of
Algorithm \ref{alg:hyp_clearing_price} with the residual multiset
$O_{{\tau^\op}_-}^{\star}$ in place of the pre-action multiset $O_n$, while the earlier snapshots $O_1,\ldots,O_{n-1}$ are left unchanged. Denote the resulting slopes by
$\hat K_{\star}^k$ and let
\[\mathcal K_{\star} = \left\{ k: O_{{\tau^\op}_-}^{\star} \text{ contains positive volume at $\alpha k$ and $\hat K_{\star}^k>0$}\right\}. \]
At the auction opening, we assume that the initial exogenous auction book contains
the carry-over schedules $g_{k,\tau^\op}^{\star}(p) = \hat K_{\star}^k(p-\alpha k)$ for $k\in\mathcal K_{\star}$. These schedules are included in $M_{\tau^\op}$ before any auction-specific arrivals are processed. Let $D^\star=\sum_{k\in\mathcal K_\star}\hat K_\star^k$ be the aggregate carry-over slope. If
$D^\star<D_\mu$, for some numerical lower slope bound $D_\mu >0$, we add the ``persistent'' fallback schedule $g^0(p) = (D_\mu-D^\star) \left(p-S_{\tau^\op}^{\mathrm{mid}}\right)$ which cannot be canceled. If
$D^\star\geqslant D_\mu$, one carry-over schedule with maximal
slope is designated as ``persistent''.

\begin{remark}
    This construction is a linear representation of the residual exogenous CLOB book which does not preserve the original side of each individual CLOB order. In fact, each retained price level becomes the zero of a two-sided linear net-supply schedule centered at that level.
\end{remark}

\section{Learning Problem for Optimal Liquidation with Closing Auctions}\label{sec:MDP}

\subsection{Problem Formulation}

The problem is discretized as follows.

\begin{definition}
\label{def:timing}
Set $t_0=0$. For every $i\geqslant0$ such that
$t_i<{\tau^\op}-1$ and every $\zeta\in\{+,-\}$, define $\tau_i^\zeta = \inf\{ s>t_i: N_s^\zeta>N_{t_i}^\zeta\}$ and $\tau_i:=\tau_i^+\vee\tau_i^-$ with the convention $\inf\emptyset=+\infty$. We set $t_{i+1} = ({\tau^\op}-1) \wedge \left[ \tau_i\vee(1+\lfloor t_i\rfloor) \right]$ until the bound ${\tau^\op}-1$ is attained. The final CLOB decision index is $n=\inf\{i\geqslant 0:t_i={\tau^\op}-1\}$. Thus there are $n+1$ CLOB decisions, $t_0,\ldots,t_n$, and we set $t_{n+1}=\tau^\op$ and $t_{n+1+\ell}=\tau^\op+\ell$ for $\ell\in \llbracket 0,h \rrbracket$. Finally, $m=n+h$, so that $t_{m+1}=\tau^\cl$, and $\mathcal T=(t_i)_{0\leqslant  i\leqslant  m+1}$.
\end{definition}

\begin{remark}
    Except for the final capped interval ending at ${\tau^\op}-1$, the maximum $\tau_i^+\vee\tau_i^-$ means that a new decision occurs only after both sides have received at least one additional market order. For simulation, the realized grid is a functional of the exogenous order-flow path. At decision $i$, however, only $t_0,\ldots, t_i$ and information available at $t_i$ are revealed to the agent. The last CLOB decision occurs at the deterministic physical time $t_n = {\tau^\op}-1$, while the random index $n$ determines how many earlier CLOB decisions occurred.
\end{remark}

{In the following, $i$ denotes the decision index whereas $t_i$ denotes
the corresponding physical time. We do not identify these two quantities
during the continuous phase, since $(t_i)_{0\leqslant i\leqslant n}$ is random. Note however that during the auction phase the two are related explicitly by $t_{n+1+\ell}=\tau^\op+\ell$ for $\ell \in \llbracket 0,h \rrbracket$.
For simplicity, with a slight abuse of notation, when there is no ambiguity, we write a generic decision time as $t$.}\vspace{0.5em}

The latent state is denoted by $X$ and defined for any time $t\in \mathcal{T}$ as a tuple 
\[X_t = (t, X_t^1, X_t^2, X_t^3, X_t^4, X_t^5, X_t^6, X_t^7, X_t^8, X_t^9, X_t^{10}, X_t^{11}, X_t^{12}, X_t^{13}, X_t^{14}, X_t^{15}, X_t^{16}, X_t^{17}, X_t^{18}),\] where each attribute encodes one of the market characteristics before the choice of an action by the agent as detailed below.

\paragraph{State Variables}
\begin{itemize}
    \item Inventory: $X_t^1=I_t$ represents the inventory of the liquidator at time $t$. By construction, $I_t\geqslant0$ for $t<\tau^\op$; inventory is frozen for $\tau^\op\leqslant t<\tau^\cl$, and $I_{\tau^\cl}\in\mathbb R$;
    \item Signed net auction fill at the clearing: $X_t^2 =0$ for $t<\tau^{\cl}$ with $X_{\tau^{\cl}}^2=Z_{\tau^{\cl}}$ and $H_{\tau^\cl}^\cl = S_{\tau^\cl}^\cl$; 
    \item {Calibrated forecast/indicative-price history: $X_t^3=(H_s^\cl)_{\mathcal T\ni s\leqslant t}$ records the calibrated CLOB forecasts of Section \ref{sec:proj} for $s<\tau^\op$, the inherited calibrated forecast $H_{\tau^\op}^\cl=H_{t_n}^\cl$ at opening, and the subsequent lagged, {tick-projected} indicative prices from \eqref{eq:hyp_cl_auction_estimation}. The raw estimator $H_s^\circ$ is internal to Algorithm \ref{alg:hyp_clearing_price};}
    \item Limit order book depth: $X_t^4 = L_t^+ \mathbf{1}_{\{0 \leqslant t \leqslant {\tau^\op}-1\}}$ and $X_t^5 = L_t^- \mathbf{1}_{\{0 \leqslant t \leqslant {\tau^\op}-1\}}$ represent the depths of the limit-order book on the ask and bid sides, respectively;
    \item Number of limit orders in the auction: $X_t^6 = M_t\mathbf{1}_{\{t \geqslant  \tau^\op\}}$;
    \item {Number of accepted market-order arrivals in the auction: $X_t^7=X_t^8=0$ for $t<\tau^\op$, while $X_t^7=N_t^+-N_{{\tau^\op}_-}^+$ and $X_t^8=N_t^--N_{{\tau^\op}_-}^-$ for $t\geqslant\tau^\op$ count buy and sell arrivals, respectively, including subsequently canceled orders;}
    \item Cancellation history: $X_t^9 = \theta_t \mathbf{1}_{\{t \geqslant  \tau^\op\}}$ represents the vector of canceled orders in the auction up to time $t$;
    \item Midprice: $X_t^{10} = S_t^\mathrm{mid}$;
    \item Auction market-order volumes on the buy and sell sides, respectively: $X_t^{11}=X_t^{12}=\emptyset$ for $t<\tau^\op$, and for $t\geqslant\tau^\op$, $X_t^{11}=(\nu_t^{+,i})_{N_{{\tau^\op}_-}^+<i\leqslant N_t^+},$ and $X_t^{12}=(\nu_t^{-,i})_{N_{{\tau^\op}_-}^-<i\leqslant N_t^-}$;
    \item Volume in the limit order book (ask/bid side): $X_t^{13} = (V_t^{+,j}\mathbf{1}_{\{t \leqslant {\tau^\op}-1\}})_{1 \leqslant j \leqslant L_t^+}$ and $X_t^{14} = (V_t^{-,j}\mathbf{1}_{\{t \leqslant {\tau^\op}-1\}})_{1 \leqslant j \leqslant L_t^-}$ record the volumes available at each depth on the ask and bid sides;
    \item Limit orders in the auction: $X_t^{15} = (g_{i,t})_{1 \leqslant i \leqslant M_t}$, with $X_t^{15} = \emptyset$ before the auction;
    \item Price history in the auction: $X_t^{16}=S^a(t)$. Set $S^a(t)=\mathbf0\in\mathbb R^h$ for $t\leqslant\tau^\op$. For $t=t_j$ with $j\in\{n+2,\ldots,m+1\}$, define $S^a(t):=(S_{t_{n+1}}^a,\ldots,S_{t_{j-1}}^a,0,\ldots,0)\in\mathbb R^h$, which records the limit-order prices submitted in the auction before time $t$;
    \item Supply/demand slope history: $X_t^{17}=K^a(t)$. Set $K^a(t)=\mathbf0\in\mathbb R^h$ for $t\leqslant\tau^\op$. For $t=t_j$ with $j\in\{n+2,\ldots,m+1\}$, define $K^a(t):=(K_{t_{n+1}}^a,\ldots,K_{t_{j-1}}^a,0,\ldots,0)\in\mathbb R^h$, which records the slopes of the limit orders submitted before time $t$;
    \item Decision index: $X_t^{18} = i$ if $t = t_i$ is the time of the $i$-th decision: let $i_t$ denote the decision index corresponding to time $t$ for a generic time $t\in \mathcal{T}$.
\end{itemize}

\begin{remark}
\label{rem:linear}
    In the linear specification of Proposition \ref{prop:linear} and in the numerical simulator, $g_{i,t}(p) = K_t^i(p-S_t^i)$ so $X_t^{15}$ can equivalently be represented by $((K_t^i, S_t^i))_{1\leqslant i \leqslant M_t}$.
\end{remark}

The state vector $X_t$ is viewed as an element of the state space $\mathcal{X}$ induced by the above variables, that is
\[\mathcal{X} = [0,\tau^\cl]\times  \mathbb{R}^2\times \Psi(\mathbb{R}_+)\times \mathbb{N}^5 \times \{0,1\}^h \times \alpha\mathbb{N}\times \Psi(\mathbb{R}_+)^4 \times \Psi(\mathcal{G}) \times  (\alpha \mathbb{N})^h\times \mathbb{R}_+^h\times \mathbb{N}\]
where
\[\mathcal{G} = \left\{ g \in \mathcal{C}(\mathbb{R}_+):g \text{ is strictly increasing}, \lim_{p\to+\infty}g(p)=+\infty\right\}\]
is the space of functions that satisfy the assumptions of Theorem \ref{th:clearing}, and for a set $E$, the space $\Psi(E)$ is defined as
\[\Psi(E) = \bigsqcup_{k=0}^{+\infty} E^k\]
with $E^0 = \{\emptyset\}$.

\begin{remark}
    While exchanges like the NYSE usually only display statistics on the total imbalance between buy and sell orders combined with an indicative price, they do not display all the orders that have been submitted at each price.
\end{remark}

Theorem \ref{th:clearing} permits general exogenous schedules. From this point onward, the learning problem and the numerical simulator use the linear specification of Proposition \ref{prop:linear} and Remark \ref{rem:linear}. Accordingly, the aggregate slope and intercept features used below are well defined. The tuple $X_t$ records the market configuration used in this formulation, while the learning algorithms receive only the reduced feature vector $\varphi(X_t)$. The simulator additionally retains internal memory needed to generate future prices and update the projected-clearing estimator, including the rough Heston history or selected historical session and the running volume moments of Algorithm \ref{alg:hyp_clearing_price}. We do not claim that either the configuration $X_t$ or its feature vector $\varphi(X_t)$ is Markov sufficient. The reduced state is $\varphi(X_t)$ with
\begin{align*}
    \varphi(X_t)&= (t,I_t, H_t^\cl,S_t^\mathrm{mid},i_t, L_t^+\mathbf{1}_{\{t < \tau^\op\}}, L_t^-\mathbf{1}_{\{t < \tau^\op\}}, V_t^{+,1}\mathbf{1}_{\{t < \tau^\op\}}, V_t^{-,1}\mathbf{1}_{\{t < \tau^\op\}},\\
    &\quad M_t\mathbf{1}_{\{t \geqslant \tau^\op\}}, {X_t^7}, {X_t^8}, C(X_t)\mathbf{1}_{\{t \geqslant \tau^\op\}}, A_{i_t}\mathbf{1}_{\{t \geqslant \tau^\op\}}, B_{i_t}\mathbf{1}_{\{t \geqslant \tau^\op\}}, \\
    &\quad G_t\mathbf{1}_{\{t \geqslant \tau^\op\}}, \mathcal{I}_t \mathbf{1}_{\{t \geqslant \tau^\op\}}, {R_t^e}\mathbf{1}_{\{t \geqslant \tau^\op\}})
\end{align*}
with $A_{i_t}$ and $B_{i_t}$ given  by the expressions of $A_j$ and $B_j$ in Section \ref{sec:clearing}, $C(x) = \max_{1 \leqslant i \leqslant m-n}\{(1-x^9_i)\mathbf{1}_{\{x_i^{17}>0\}}\}$ for a state $x\in \mathcal{X}$, and
\[G_t = \sum_{i=1}^{M_t} K_t^i, \quad {R_t^e} = \sum_{i=1}^{M_t} K_t^i S_t^i, \quad \mathcal{I}_t = \sum_{i={N_{{\tau^\op}_-}^++1}}^{N_t^+} \nu_t^{+,i}-\sum_{i={N_{{\tau^\op}_-}^-+1}}^{N_t^-} \nu_t^{-,i}.\]
Let $\mathcal{Y} = \varphi(\mathcal{X})$.

\paragraph{Action Variables} We now define the actions of the strategic trader. Given a state vector $X_t\in \mathcal X$, we define the action vector $A_t$ as $A_t = (A_t^1, A_t^2, A_t^3, A_t^4, A_t^5)$ where each component represents a particular action.
\begin{itemize}
    \item Volume set in the limit order: $A_t^1 = v_t \mathbf{1}_{\{0 \leqslant t \leqslant {\tau^\op}-1\}}$;
    \item Depth offset in the limit order book: $A_t^2 =\delta_t\mathbf{1}_{\{0 \leqslant t \leqslant {\tau^\op}-1\}}$;
    \item Supply/demand slope and local reference-price offset: $A_t^3 =K_t^a\mathbf{1}_{\{\tau^\op \leqslant t \leqslant \tau^\cl-1\}}$;
   and $A_t^4 = \ell_t^a\mathbf{1}_{\{\tau^\op \leqslant t \leqslant \tau^\cl-1\}}$ such that $S_t^a = S_{\tau^\op}^\mathrm{mid}+\alpha b_t^a$ with \[{b_t^a=\bigl(b^{\mathrm{anc}}(\varphi(X_t))+\ell_t^a\bigr)\mathbf1_{\{K_t^a>0\}}};\]
    \item Order cancellation in the auction: $A_t^5 =c_t\mathbf{1}_{\{\tau^\op \leqslant t \leqslant \tau^\cl-1\}}$.
\end{itemize}
\begingroup
For the selected anchor convention, define the integer anchor offset for a raw observation $y$ by
\[
b^{\mathrm{anc}}(y)=\begin{cases}
\displaystyle\left\lfloor\frac{y_3-y_4}{\alpha}+\frac12\right\rfloor,
&\text{indicative anchor},\\[0.5em]
0,&\text{frozen-mid anchor}.
\end{cases}
\]
During the auction, $y_4=S_{\tau^\op}^{\mathrm{mid}}$ and $y_3=H_t^\cl$, including the inherited calibrated CLOB forecast at opening. Under the frozen-mid anchor, a positive-slope schedule instead has $b_t^a=\ell_t^a$ and $S_t^a=S_{\tau^\op}^{\mathrm{mid}}+\alpha\ell_t^a$. Zero slope uses the canonical $\ell_t^a=b_t^a=0$; its reference price has no effect on the book. The same anchor offset is used in the action mapping and admissibility conditions below, including the absolute price-deviation and nonnegative-price constraints.
\endgroup

We view such actions as part of the ambient action set
\[
\begin{aligned}
\mathcal A =
\llbracket0,V_{\max}\rrbracket \times \llbracket0,L_{\max}\rrbracket \times \beta \llbracket0,K_{\max}\rrbracket \times \llbracket-B_\text{max},B_\text{max}\rrbracket \times \{0,1\}
\end{aligned}
\]
where $\beta>0$ is the auction slope increment and $V_\text{max}, L_\text{max}, B_\text{max}, K_\text{max} \in \mathbb{N}^\star$ are numerical bounds defining the finite action grid.\vspace{0.5em}

During continuous trading, the agent liquidates her inventory by submitting a volume $A_t^1\leqslant I_t=X_t^1$ at a price $S_t^{\mathrm{mid}}+\alpha A_t^2\geqslant S_t^{\mathrm{mid}}$. At an auction decision, $c_t=1$ is admissible if and only if at least one prior positive-slope strategic schedule is live. It cancels every such schedule submitted strictly before $t$ and the current schedule is inserted afterward. As $\theta_t^{(s)}=1$ if and only if the positive-slope schedule submitted in slot $s$ has already been canceled (with $\theta_t^{(s)}=0$ if that slot contained no order), we obtain the admissibility condition $A_t^5 = c_t \leqslant \max_{1\leqslant i \leqslant t-\tau^\op}(1-\theta_t^{(i)})\mathbf{1}_{\{K_{t_{n+i}}^a>0\}} =:C_t$ with the convention that $\max\emptyset = 0$ (for $t\leqslant \tau^\op$).
\begin{definition}[Admissible actions]
    Given an observed state $y \in \mathcal{Y}$, the set $\operatorname{Adm}(y)$ of admissible actions is defined as
    \begin{align*}
        \operatorname{Adm}(y) &= \{a\in \mathcal{A}: a_1 \leqslant y_2\mathbf{1}_{\{y_1<\tau^\op\}}, a_5  \leqslant y_{13}, y_4 +\alpha ({b^{\mathrm{anc}}(y)}+a_4)\mathbf{1}_{\{a_3>0\}} \geqslant 0,\\ &\quad a_1=0\implies a_2 = 0, a_3=0\implies a_4 = 0, (a_1,a_2)\mathbf{1}_{\{y_1\geqslant\tau^\op\}}=0, \\ &\quad(a_3,a_4,a_5)\mathbf{1}_{\{y_1<\tau^\op\}} = 0,
        |a_4+{b^{\mathrm{anc}}(y)}|\mathbf{1}_{\{a_3>0\}}\leqslant B_\infty\}
    \end{align*}
    where $B_\infty>0$ and the two indicator constraints impose phase compatibility.
\end{definition}

\begin{remark}
    Note that $S_t^\mathrm{mid}$ is constant during the auction. Furthermore, since all $X_t^1$, $X_t^{10}$ and $C(X_t)$ are encoded in $\varphi(X_t)$, all admissible actions for $X_t$ are exactly those admissible for $\varphi(X_t)$.
\end{remark}

\begin{definition}[Admissible policies]
    An admissible policy is a measurable map \[\pi\colon y\in \mathcal{Y} \longmapsto \pi(\cdot \mid y) \in \mathcal{P}(\mathcal{A})\]
    such that $\pi(\operatorname{Adm}(y)\mid y)=1$ for all $y \in \mathcal{Y}$. Let $\Pi$ denote the set of these admissible policies.
\end{definition}

\paragraph{Reward} {The weighted preference objective $J$ uses the rewards below, before initial-value centering and learning-only conditioning. The fictive auction-credit weight is $\omega_A\geqslant0$ (we use $10^{-4}$ in the simulations). It weights interim credit and terminal purchase attenuation, but not cash, fees, residual marking or inventory risk. The CLOB multiplier $f^c$ is a separate preference.} {The rewards use the calibrated or lagged $H_t^\cl$ defined in Section~\ref{sec:proj}, which is the last component of $X_t^3$. The auction anchor changes $S_t^a$.}

\begin{enumerate}
    \item During the continuous trading session ($t<\tau^\op$): The trader submits a price $S_t^\bullet=S_t^\text{mid}+\alpha A_t^2$. The volume executed is given by $E_t$. The cash proceeds are thus given by $S_t^\bullet\times E_t$. We moreover assume that the agent penalizes the execution by comparing the execution price with the hypothetical clearing price $H_t^{\cl} = \varphi(X_t)_3$. If $S_t^\bullet>H_t^{\cl}$, {the cash proceeds enter the reward without attenuation}; if $S_t^\bullet<H_t^{\cl}$, the agent may regret the execution. {The tolerated shortfall of the execution price $S_t^\bullet$ below the current scalar forecast $H_t^\cl$ is $k^\star\alpha$, for fixed $k^\star>0$.} In the numerical specification, the CLOB reward multiplier is strictly positive when
\[(H_t^{\cl}-S_t^\bullet)_+<k^\star\alpha,\]
and vanishes when this gap is at least $k^\star\alpha$. It equals one when $S_t^\bullet\geqslant H_t^{\cl}$. We thus introduce a penalty function $f^c:\mathbb R\longrightarrow [0,1]$ continuous and nondecreasing such that $f^c$ is zero on $\mathbb{R}_-$ and positive on $(0,k^\star \alpha]$ with $f^c(k^\star\alpha)=1$. The resulting reward is 
\[r_t(X_t,A_t)= S_t^\bullet E_t f^c(k^\star \alpha -( H_t^\cl -S_t^\bullet)_+).\]
\item During the auction trading session ($\tau^\op\leqslant t<\tau^\cl$): The trader submits a slope $K_t^a = A_t^3$ and a price
\[{S_t^a = S_{\tau^\op}^\mathrm{mid}+\alpha b_t^a,\quad b_t^a=\bigl(b^{\mathrm{anc}}(\varphi(X_t))+A_t^4\bigr)\mathbf1_{\{K_t^a>0\}}.}\]
\begingroup
Purchase attenuation uses a convex, continuous, nonincreasing function $f^a:\mathbb R\to\mathbb R_+$, with $f^a(x)=0$ for $x\geqslant0$ and $f^a(x)\leqslant(-x)_+$. The weighted schedule credit and auction reward are
\[
u_t=\omega_A\left[K_t^aH_t^\cl(H_t^\cl-S_t^a)
+f^a\!\left(K_t^aH_t^\cl(H_t^\cl-S_t^a)\right)\right],
\]
and
\[
r_{t_j}(X_{t_j},A_{t_j})=u_{t_j}-c_{t_j}\left[d_{t_j}+
\sum_{i=n+1}^{j-1}\mathbf1_{\{K_{t_i}^a>0\}}(1-\theta_{t_j}^{(i-n)})u_{t_i}\right].
\]
The cancellation charge is $d_{t_j}c_{t_j}$, with a rate determined by elapsed auction time. The fictive schedule credits are noncash preference terms in $J$, distinct from the policy-invariant learning conditioning of Section \ref{subsec:conditioning}.
\endgroup
\item At clearing ($t=\tau^\cl$): At the clearing time, the clearing price $S_{\tau^\cl}^\cl$ is determined and orders are matched. {The signed execution cash is $S_{\tau^\cl}^\cl Z_{\tau^\cl}$ for the allocation of her live schedules.} The agent is furthermore penalized for holding inventory and the agent's reward associated with purchases is once again attenuated. We introduce $\lambda > 0$ as a penalization parameter. Since the midprice is frozen throughout the auction phase in our model, any inventory remaining after the auction is priced at the last exogenous midprice $S^\text{mid}_{\tau^\op}$ rather than at the auction clearing price. This separates the valuation of residual inventory from the endogenous auction price, which may be affected by the agent's own submitted schedules. The final reward of the trader is given by
\begin{align*}
    G(X_{\tau^\cl}) &= S_{\tau^\cl}^\cl Z_{\tau^\cl}- \lambda |I_{\tau^\cl}|^2 + {\omega_A} f^a\left (S_{\tau^\cl}^\cl Z_{\tau^\cl}\right) + S^\text{mid}_{\tau^\op}I_{\tau^\cl}.
\end{align*}
\end{enumerate}

\begin{remark}
    {Numerically, $f^c(t)=\min\{1,t_+/(k^\star\alpha)\}$ and $f^a(t)=q(-t)_+$ for $q\in[0,1]$. The ``headline'' has $q=0$: purchase attenuation disappears, but the CLOB multiplier, weighted interim credit and exact cancellation clawback remain active.}
\end{remark}

During the CLOB phase, the agent chooses action $A_{t_i}$ in state $X_{t_i}$. Then, the executed volume $E_{t_i}$ is realized, determining the one-step reward $r_{t_i}(X_{t_i},A_{t_i})$ during the state transition. Finally, the agent transitions into state $X_{t_{i+1}}$.\vspace{0.5em}

{The ``headline'' training objective is the weighted preference objective $J$, defined by the rewards $r_t$ and $G$ above:}
\begin{align*}
    (\mathbf{P})\quad \sup_{\pi \in \Pi} \quad & J(\pi) = \mathbb{E}_\pi\left [\sum_{i=0}^{m} r_{t_i}(X_{t_i},A_{t_i}) + G(X_{\tau^\cl}) \right] \nonumber \\
    \text{s.t.} \quad & \left \{\begin{array}{l}
    X_0 \text{ given}  \\
    A_{t_i} \sim \pi(\cdot \mid \varphi(X_{t_i}))
\end{array}\right..
\end{align*}

\begingroup
\subsection{Implementation Shortfall and Evaluation Criteria}
\label{subsec:shortfall}
The signed auction fill is $Z_{\tau^\cl}$, its clearing price is $S_{\tau^\cl}^{\cl}$, and residual inventory is $I_{\tau^\cl}=I_{\tau^\op}-Z_{\tau^\cl}$. Relative to immediate execution at the initial midprice, episode implementation shortfall is
\begin{equation}
\mathrm{IS}=S_0^{\mathrm{mid}}I_0-\left[\sum_{i=0}^{n}S_{t_i}^{\bullet}E_{t_i}
+S_{\tau^\cl}^{\cl}Z_{\tau^\cl}
+S_{\tau^\op}^{\mathrm{mid}}I_{\tau^\cl}-\sum_{i=n+1}^{m}d_{t_i}c_{t_i}\right].
\label{eq:shortfall}
\end{equation}
Positive shortfall denotes underperformance of the initial-midprice benchmark; negative shortfall denotes improvement, and lower is better. It is the negative of the marked-to-market economic PnL defined by the same cash flows. We value residual inventory at the frozen auction-opening midprice. The separate inventory preference gives
\begin{equation}
\mathrm{IS}_\lambda=\mathrm{IS}+\lambda I_{\tau^\cl}^{2}.
\label{eq:penalized_shortfall}
\end{equation}
The quadratic term is neither a cash expense nor part of ordinary implementation shortfall.

\begin{remark}
We call this objective $\mathrm{IS}_\lambda$ the economic objective and we will use the prefix E in the numerical simulation when referring to it. Economic-only training replaces the CLOB multiplier by one and omits interim credits, their clawbacks and terminal purchase attenuation, while retaining cash, fees, residual marking and terminal inventory risk. These preferences change $J$. We use the weighted $J$ for training, and the expected inventory-penalized shortfall for validation checkpoint selection and policy comparisons.
\end{remark}
\endgroup

\section{Learning Algorithms}
\label{sec:learning}
\label{app:implementation}
\begingroup
\subsection{Numerical Representation and Learning Conditioning}
\label{subsec:conditioning}
The policy receives the common 18-coordinate raw observation $y_t=\varphi(X_t)$. Network inputs $\tilde y_t=\mathcal N(y_t)$ replace absolute forecast prices by $H_t^\cl-S_t^{\mathrm{mid}}$ and the two weighted-quote coordinates by continuous-root residual exposure and displacement. Define
\[
\begin{aligned}
W_t^a&=B_{i_t}-S_t^{\mathrm{mid}}A_{i_t},&
W_t^e&={R_t^e}-S_t^{\mathrm{mid}}G_t,\\
\Delta_t&=\frac{W_t^a+W_t^e+\mathcal I_t}{A_{i_t}+G_t},&
\hat Z_t&=A_{i_t}\Delta_t-W_t^a.
\end{aligned}
\]
For positive aggregate slope, the replacement coordinates are $I_t-\hat Z_t$ and $\Delta_t$; otherwise they are zero. During the auction, \[\Delta_t = \frac{W_t^a+W_t^e+\mathcal I_t}{A_{i_t}+G_t} =\frac{B_{i_t}+R_t^e+\mathcal I_t}{A_{i_t}+G_t} -S_t^{\mathrm{mid}}
= p_t^\star-S_t^{\mathrm{mid}}\] where $p_t^\star$ is the continuous clearing root of the current book. Thus the common input representation supplies current clearing-root information independently of the explicit $H_t^\cl$ feature. Disabling that feature retains the root-derived coordinates and the remaining aggregate book features. DQN, DDPG and TD3 use phase-specific normalization and signed-asinh inventory transforms; SAC uses pooled linear normalization.  Specifically, DQN, DDPG and TD3 apply $x\mapsto\operatorname{asinh}(x/I_s)$, with $I_s=\alpha/\lambda$, to auction inventory and continuous-root residual inventory when aggregate slope is positive. CLOB inventory remains linear; SAC omits this nonlinear transform.\vspace{0.5em}

Write $y_i=\varphi(X_{t_i})$ at decision indices. Let $R_i^0=r_{t_i}(X_{t_i},A_{t_i})+\mathbf1_{\{i=m\}}G(X_{\tau^\cl})$. Initial-value centering yields
\[
R_i=R_i^0-S_0^{\mathrm{mid}}\left[\mathbf1_{\{i\leqslant n\}}E_{t_i}
+\mathbf1_{\{i=m\}}(Z_{\tau^\cl}+I_{\tau^\cl})\right].
\]
Inventory conservation makes the total subtraction $S_0^{\mathrm{mid}}I_0$, independent of policy. For each setting, we generate 12 training-split inventory paths using uniformly sampled admissible CLOB actions. We linearly interpolate each path, including its values at $t=0$ and $t = \tau^\op$, onto the one-minute grid. We average the interpolated inventories pointwise and define $q_\text{ref}(t)$ by linear interpolation of these grid-point means. The resulting curve is frozen before training and shared across learners and treatments, with $q_\text{ref}(0) = I_0 = 100$. The learning reward evaluates it at each interval midpoint:
\[
\tilde R_i=R_i+\Phi(y_{i+1})-\Phi(y_i)
-\mathbf1_{\{i\leqslant n\}}q_{\mathrm{ref}}\!\left(\frac{t_i+t_{i+1}}2\right)
\left(S_{t_{i+1}}^{\mathrm{mid}}-S_{t_i}^{\mathrm{mid}}\right),
\quad\Phi(y_{m+1})=0.
\]
The last CLOB interval is included in the reference adjustment. The potentials are evaluated before input normalization:
\[
\Phi(y_t)=\begin{cases}
(S_t^{\mathrm{mid}}-S_0^{\mathrm{mid}})I_t
-\lambda\left[I_t-I_0(1-t/\tau^\op)\right]_+^2,&t<\tau^\op,\\
(S_{\tau^\op}^{\mathrm{mid}}-S_0^{\mathrm{mid}})I_t
+\Delta_t\hat Z_t-\lambda(I_t-\hat Z_t)^2,&\tau^\op\leqslant t<\tau^\cl.
\end{cases}
\]
Sparse auction credit omits $\Delta_t\hat Z_t-\lambda(I_t-\hat Z_t)^2$, leaving a constant inventory-mark potential between opening and settlement. All schemes retain terminal marking and risk; their enabled conditioning terms are listed in Table \ref{tab:treatment_specification}.\vspace{0.5em}

For either auction potential, summation gives
\[
\sum_{i=0}^{m}\tilde R_i
=\sum_{i=0}^{m}R_i-\Phi(y_0)
-\sum_{i=0}^{n}q_{\mathrm{ref}}\!\left(\frac{t_i+t_{i+1}}2\right)
\left(S_{t_{i+1}}^{\mathrm{mid}}-S_{t_i}^{\mathrm{mid}}\right).
\]
Here $\Phi(y_0)=0$. Conditional on a fixed fitted reference curve, the last sum is policy-independent on common exogenous midprice paths and control-independent decision times. With a common initial state and integrable rewards, these adjustments preserve policy differences pathwise and expected rankings. For completeness, in the ``headline'', $\mathbb E_\pi[\sum_iR_i]=J(\pi)-S_0^{\mathrm{mid}}I_0$. Under centered economic training, $\sum_iR_i=-\mathrm{IS}_\lambda$; without centering, the economic reward sum is $S_0^{\mathrm{mid}}I_0-\mathrm{IS}_\lambda$.

\subsection{Deep Q-Network}
\label{subsec:dqn}
DQN uses separate CLOB and auction networks with shared action encoders within each phase. A normalized ReLU state trunk produces an unrestricted scalar $V_\kappa$ and query $u_\kappa\in\mathbb R^r$; a normalized encoder maps action coordinates to a bounded embedding $e_\kappa(a)$. The two architectures differ:
\[
\begin{aligned}
Q_{\theta_c}(\tilde y,a)&=V_c(\tilde y)
+\frac{\langle u_c(\tilde y),e_c(a)\rangle}{\sqrt r}
+\omega_c\mathbf1_{\{a\ne a_c^0\}},\\
Q_{\theta_a}(\tilde y,a)&=V_a(\tilde y)
+\frac{\langle u_a(\tilde y),e_a(a)-e_a(a_a^0)\rangle}{\sqrt r}.
\end{aligned}
\]
Here $a_\kappa^0$ is the phase no-op. The CLOB coefficient $\omega_c$ is learned from zero; the auction non-no-op coefficient is fixed at zero, so its no-op value equals $V_a$. Other weights use ordinary initialization. The DQN action coordinates are $(v/V_{\max},\delta/L_{\max})$ and $(K^a/(\beta K_{\max}),\ell^a/B_{\max},c)$. Each state hidden layer has non-affine LayerNorm before ReLU. The action encoder has two width-$r$ linear layers, with ReLU after the first and non-affine LayerNorm followed by $\tanh$ after the second. Targets initially copy the online networks.\vspace{0.5em}

Training uses one-step masked Double DQN \cite{mnih2015human,van2016deep}. Nonterminal targets use the next phase's networks, including the auction continuation at the CLOB boundary; terminal targets contain no bootstrap. Raw observations determine masks and potentials, while replay stores transformed observations and conditioned rewards. The common structured-exploration period is followed by masked $\varepsilon$-greedy behavior; auction exploration mixes zero-slope controls with all admissible actions. Table \ref{tab:hyperparams} gives the hyperparameters. Algorithm \ref{alg:dqn} specifies updates and checkpoint selection for this architecture.

\begin{algorithm}[htbp!]

\caption{Phase-coupled masked Double DQN}
\label{alg:dqn}
\begin{algorithmic}[1]
\REQUIRE Hyperparameters in Table \ref{tab:hyperparams}, frozen input transform, and selection protocol in Section \ref{subsec:experimental_protocol}
\STATE Initialize online/target phase networks, optimizers and replay buffers
\FOR{each training episode permitted by the stopping rule}
  \FOR{each decision, with raw observation $y$ in phase $\kappa$}
    \STATE Select a feasible action using structured exploration or masked $\varepsilon$-greedy behavior
    \STATE Observe $y'$, terminal indicator $\eta$ and reward; compute $\tilde R$ using raw observations
    \STATE Store $(\mathcal N(y),a,\tilde R,\mathcal N(y'),\eta,\kappa',\operatorname{Adm}(y'))$ in the current phase buffer
    \IF{structured exploration is complete and the current buffer meets its threshold}
      \STATE Sample a minibatch; for each row set $Y=\tilde R$ if terminal, otherwise
      \STATE \hspace{1em}$a^*=\arg\max_{a'\in\operatorname{Adm}(y')}Q_{\theta_{\kappa'}}(\mathcal N(y'),a')$,\quad $Y=\tilde R+Q_{\theta_{\kappa'}^-}(\mathcal N(y'),a^*)$
      \STATE Take one AdamW step on the mean squared TD error, with cleared gradients and the tabulated norm bound
      \STATE Update only the current phase target: $\theta_\kappa^-\gets(1-\tau_{\mathrm{targ}})\theta_\kappa^-+\tau_{\mathrm{targ}}\theta_\kappa$
    \ENDIF
  \ENDFOR
  \STATE Apply eligible joint validation and stopping
\ENDFOR
\RETURN The joint checkpoint with the lowest eligible mean inventory-penalized shortfall
\end{algorithmic}
\end{algorithm}

\subsection{Projected Continuous-Proposal Control}
DDPG, TD3 and SAC \cite{lillicrap2020continuous,fujimoto2018addressing,haarnoja2018soft,achiam2018spinningup} propose controls in $[-1,1]^2$ during CLOB trading and $[-1,1]^3$ during the auction. The environment executes $\Gamma_y(u)$, while replay retains the box-clipped proposal $u$. The projection uses half-up rounding, inventory-capped CLOB sales, anchor-dependent auction price bounds and cancellation eligibility. While we use $J$ for training, and $\mathrm{IS}_\lambda$ for validation checkpoint selection and policy comparisons, SAC additionally uses phase-specific, automatically tuned proposal-entropy regularization.\vspace{0.5em}

Write $\mathsf r(x)=\lfloor x+1/2\rfloor$. After clipping a proposal $u$ to its box, CLOB execution uses
\[
v=\min\{\lfloor I_t\rfloor,V_{\max},\mathsf r((u_1+1)V_{\max}/2)\},\quad
\delta=\begin{cases}0,&v=0,\\\mathsf r((u_2+1)L_{\max}/2),&v>0.\end{cases}
\]

Auction proposals give $K^a=\beta\mathsf r((u_1+1)K_{\max}/2)$, $\ell_0=\mathsf r(B_{\max}u_2)$ and $c=\mathbf1_{\{u_3\geqslant0\}}C(X_t)$. For positive slope, define the local-offset set
\[
\mathcal{L}(y) =\left\{\ell\in\llbracket-B_{\max},B_{\max}\rrbracket:
|b^{\mathrm{anc}}(y)+\ell|\leqslant B_\infty,\quad
S_{\tau^\op}^{\mathrm{mid}}+\alpha(b^{\mathrm{anc}}(y)+\ell)\geqslant0\right\},
\]
and clip $\ell_0$ to the endpoints of $\mathcal L(y)$ when this set is nonempty. Otherwise, we set $K^a=\ell=0$. Zero slope always uses $\ell=0$. Thus $\Gamma_y(u)\in\operatorname{Adm}(y)$ for every box-clipped proposal $u$.
We use Stable-Baselines3 2.7.1 with separate phase models and ordinary initialization \cite{raffin2021stable}. DDPG adds affine LayerNorm before each critic hidden-layer activation; its actor is unchanged. TD3 and SAC use twin critics without this additional normalization.\vspace{0.5em}

For continuous learners, let $\mu_a^-$ denote the auction target actor and $Q_a^-$ (or $Q_{a,j}^-$) its target critic(s). The auction continuation at a sampled CLOB junction is
\[
\begin{cases}
Q_a^-(\tilde y',\mu_a^-(\tilde y')),&\text{for DDPG},\\
\min_jQ_{a,j}^-(\tilde y',u'),&\text{for TD3},\\
\min_jQ_{a,j}^-(\tilde y',u')-\zeta_a\log\pi_a(u'\mid\tilde y'),&\text{for SAC}.
\end{cases}
\]
TD3 adds clipped Gaussian target noise to its target proposal and clips the result to the proposal box. SAC samples its current actor and uses its current positive entropy coefficient $\zeta_a$. The continuation is added at sampling time, with further CLOB bootstrapping suppressed. DDPG/TD3 actor updates maximize the first critic, while SAC includes its entropy term. The critic-loss aggregation and actor/target update periods are given in Table \ref{tab:hyperparams}. DDPG's actor rate is applied by its optimizer hook and SAC's temperature rate follows the same fixed episode decay as its network rates. Optimizer hooks impose the actor and critic gradient bounds in Table \ref{tab:hyperparams}; the SAC temperature optimizer is not clipped by these hooks. SAC adjusts entropy automatically in each phase. After common structured exploration, DDPG/TD3 use Gaussian proposal noise and SAC samples its stochastic policy. Evaluation uses deterministic actors, including the squashed Gaussian mean for SAC. All learners return the joint checkpoint selected under Section \ref{subsec:experimental_protocol}.

\begin{table}[htbp!]

\centering
\setlength{\tabcolsep}{4pt}
\renewcommand{\arraystretch}{1.04}
\begin{tabular}{lcccc}
\toprule
\textbf{Setting} & \textbf{DQN} & \textbf{DDPG} & \textbf{TD3} & \textbf{SAC} \\
\midrule
Backend & Native & \multicolumn{3}{c}{Stable-Baselines3 2.7.1} \\
Replay capacity (per phase) & \multicolumn{4}{c}{$50{,}000$} \\
Replay minima $N_w^c/N_w^a$ & \multicolumn{4}{c}{$512/512$} \\
Structured exploration & \multicolumn{4}{c}{Episodes $0$--$31$; no optimizer updates} \\
Minibatch size & \multicolumn{4}{c}{$128$} \\
Hidden widths (ReLU) & \multicolumn{4}{c}{$256,256$} \\
Replay reward scale $c_r$ & \multicolumn{4}{c}{$1$} \\
One-step continuation factor & \multicolumn{4}{c}{$1$} \\
Critic updates per transition & \multicolumn{4}{c}{At most one, after both learning gates} \\
Optimizer & AdamW & Adam & Adam & Adam \\
Weight decay & $10^{-4}$ & $0$ & $0$ & $0$ \\
Gradient-norm bound & $10$ & $1$ & $1$ & $1$ \\
Action-embedding rank & $32$ & -- & -- & -- \\
\midrule
\multicolumn{5}{l}{\emph{Initial rates; multiplied by $\max\{0.1,2^{-e/90}\}$ in episode $e$}}\\
Critic & $1.5\times10^{-4}$ & $3\times10^{-4}$ & $3\times10^{-4}$ & $3\times10^{-4}$ \\
Actor & -- & $10^{-4}$ & $3\times10^{-4}$ & $3\times10^{-4}$ \\
Temperature & -- & -- & -- & $3\times10^{-4}$ \\
\midrule
Number of online critics & $1$ & $1$ & $2$ & $2$ \\
Critic loss & MSE & MSE & Sum of MSEs & Half-sum of MSEs \\
Polyak coefficient & $0.005$ & $0.0025$ & $0.005$ & $0.005$ \\
Target-update period$^\dagger$ & $1$ & $1$ & $2$ & $1$ \\
Actor-update period$^\dagger$ & -- & $1$ & $2$ & $1$ \\
\midrule
\multicolumn{5}{l}{\emph{Behavior after structured exploration; buffer fallback specified in text}}\\
Behavior & Masked mixture & Gaussian & Gaussian & Stochastic \\
DQN epsilon endpoints & $1\to0.05$ & -- & -- & -- \\
DQN decay start/end ($e$) & $5/95$ & -- & -- & -- \\
Auction zero-slope mixture & $0.5$ & -- & -- & -- \\
Proposal-noise std. & -- & $0.1$ & $0.1$ & -- \\
TD3 target-noise std./clip & -- & -- & $0.2/0.5$ & -- \\
Initial entropy coefficient & -- & -- & -- & $0.001$ (auto) \\
Target entropy (CLOB/auction) & -- & -- & -- & $-2/-3$ \\
\bottomrule
\end{tabular}
\caption{Learner hyperparameters}
\label{tab:hyperparams}
\end{table}
\endgroup

\section{Stylized Reference Benchmarks}
\label{sec:benchmark}

We compare the learned policies against two reference benchmarks. Both benchmarks illustrate closed-form approaches; the primary comparison is among learners. The first benchmark is adapted from the optimal market making models of Avellaneda and Stoikov \cite{avellaneda2008high} later solved explicitly by Guéant, Lehalle and Fernandez-Tapia \cite{gueant2013dealing}. We adapt the market making model to liquidation only and without risk aversion for the continuous phase, and approximate it on our discrete grid. The second benchmark is the time-weighted average price policy for the continuous phase. We adopt the same heuristic liquidation rule for the auction phase. These benchmarks are external reference policies for optimal liquidation and are not required to belong to the learned policy class $\Pi$. A theoretical solution to the optimal market making problem on a closing auction is left for future work.

\subsection{Avellaneda--Stoikov Optimal Market Making}\label{section:AS}

In this section, we will recall the main results of \cite{avellaneda2008high} and \cite{gueant2013dealing}, in the case where the market maker acts as a seller only. We then propose a heuristic model for the auction phase.

\subsubsection{Optimal Market Making along the CLOB} We consider the continuous phase, that is $t < \tau^\op$ in what follows. For this auxiliary continuous-time benchmark, we work on $[0,T]$ with the chosen horizon $T=t_n$, where $t_n={\tau^\op}-1$ is the final CLOB decision time in our initial framework. The market midprice is assumed to follow an arithmetic Brownian motion $\mathrm{d}S_t^\mathrm{mid}=\sigma \mathrm{d}B_t$ with $\sigma>0$. We assume
that transactions have constant unit size. {Fix an integer inventory cap $Q\in\mathbb N$ and initial inventory $q_0=I_0\in\{0,\ldots,Q\}$. The inventory process is $q_t=q_0-N_t^a$,} where {$N^a$ is the controlled counting process giving the cumulative number of shares sold by the market maker}. Following \cite{avellaneda2008high}, we assume that the intensity of $N^a$ depends on the spread $\delta_t^a = k_t-k_t^{\mathrm{mid}}$ as {$\lambda^a_t(\delta^a_t)=\mathbf1_{\{q_{t-}>0\}}Ae^{-\alpha k\delta_t^a}$} for $A,k > 0$ calibrated heuristically from the simulator (Section \ref{sub:rough}).
{To specify the driving noise, let $\mathcal M^a$ be a Poisson random measure on $[0,T]\times\mathbb R_+$ with compensator $\mathrm ds\,\mathrm dz$, independent of $B$. Let $\mathbb F^{\mathrm{AS}}$ be the usual augmentation of the filtration jointly generated by $B$ and $\mathcal M^a$ up to time $t$. Prices and execution counts are adapted to this common filtration, and quote controls are predictable with respect to it. Executions are generated by
\[
\forall t \in [0,T],\quad N_t^a=\int_{(0,t]\times\mathbb R_+}
\mathbf1_{\{z\leqslant\lambda_s^a(\delta_s^a)\}}
\mathcal M^a(\mathrm ds,\mathrm dz).
\]
The primitive noises are independent, but the controlled count $N^a$ need not be independent of $B$ when quotes respond to the price history.}
The cash process $C_t^{\text{AS}}$ of the market maker evolves according to $\mathrm{d}C_t^{\text{AS}} = (S_t^\mathrm{mid}+\alpha\delta_t^a)\mathrm{d}N_t^a$. {Let $T=t_n$ and let $\tilde{\mathcal A}$ consist of bounded {$\mathbb F^{\mathrm{AS}}$-predictable} quote controls while $q_{t-}>0$, with no order posted and no further executions at $q=0$. The stopped intensity keeps inventory in $\{0,\ldots,Q\}$.} The market maker optimizes
\[(\mathbf{M})\quad \sup_{\delta^a \in \tilde{\mathcal A}} \mathbb{E}\left [C_T^{\text{AS}} + q_{T}S_{T}^\mathrm{mid} \right].\]

\begin{proposition}
\label{prop:benchm}
{For $t\in[0,T)$ and $q\in\{1,\ldots,Q\}$, the optimal quotes solving $(\mathbf M)$ are given by}
\[\delta^{a,\star}(t,q) = \frac{1}{\alpha k}\left [ 1+\ln \left ( \frac{v_q(t)}{v_{q-1}(t)}\right)\right],\] where 
\[\forall q \in \llbracket 0,Q \rrbracket, \quad v_q(t) = \sum_{j=0}^q \frac{(Ae^{-1}(T-t))^j}{j!}.\]
{The boundary conditions are $v_0(t)=1$ for $0\leqslant t\leqslant T$ and $v_q(T)=1$ for $0\leqslant q\leqslant Q$. At $q=0$ the policy abstains. At time $t=T$, the final quote uses the continuous extension $\delta^{a,\star}(T,q)=1/(\alpha k)$ for $q \geqslant 1$, which follows from $v_q(T)=v_{q-1}(T)=1$.}
\end{proposition}

\begin{proof}
Note that $(\mathbf{M})$ corresponds to the risk-neutral market maker $\gamma\to 0$ in the standard optimal market making problem investigated in \cite{avellaneda2008high,gueant2013dealing}. In this case, one proves that the problem reduces to solving the ODE system 
\[\forall q \in \llbracket 1,Q \rrbracket, \quad v_q'(t) = -Ae^{-1} v_{q-1}(t),\] and we directly get optimal quotes from \cite{gueant2013dealing} when $\gamma\to 0$ by \[\delta^{a,\star}(t,q) = \frac{1}{\alpha k}\left [ 1+\ln \left ( \frac{v_q(t)}{v_{q-1}(t)}\right)\right].\]

    We prove the result by induction. Note that $v_0 = 1$. For a fixed $q\in\{1,\ldots,Q\}$, assume that 
    \[v_{q-1}(t)= \sum_{j=0}^{q-1} \frac{(Ae^{-1}(T-t))^j}{j!}.\] 
    Since $v_q(T)=1$, we compute
\begin{align*}
    v_q(t) &= v_q(T)- \int_t^T v_q'(s)\mathrm{d}s =1 + Ae^{-1}\int_t^T v_{q-1}(s)\mathrm{d}s= 1 + Ae^{-1}\sum_{j=0}^{q-1} \int_t^T \frac{(Ae^{-1}(T-s))^j}{j!}\mathrm{d}s\\
    &= 1 + \sum_{j=0}^{q-1} \frac{(Ae^{-1})^{j+1}}{j!}\int_0^{T-t} s^j \mathrm{d}s= 1+\sum_{j=0}^{q-1} \frac{(Ae^{-1}(T-t))^{j+1}}{(j+1)!}= \sum_{j=0}^q \frac{(Ae^{-1}(T-t))^j}{j!}.
\end{align*}
This completes the proof by induction. The rest follows.
\end{proof}

{The quote formula applies only for $q\in\{1,\ldots,Q\}$; the zero-inventory state has the no-trading boundary specified above.} The continuous-time derivation assumes unit-sized transactions and controls the quote offset $\delta^{a,\star}(t,q)$, measured in ticks, rather than a separately submitted volume. Before grid projection, the corresponding quote price is $S_t^{\mathrm{mid}}+\alpha\delta^{a,\star}(t,q)$. At each CLOB decision time $t=t_i$, the agent submits $v_t=V_{\max}\wedge\lfloor q_t\rfloor$ units. If $v_t>0$, the quote offset $\delta_t$ is obtained by evaluating $\delta^{a,\star}(t,1\vee\lceil q_t\rceil)$ and projecting it to the nearest admissible integer offset. The submitted price $S_t^\bullet=S_t^{\mathrm{mid}}+\alpha\delta_t$ remains fixed throughout $(t_i,t_{i+1}]$: all fills of that order occur at this price, rather than at quotes recomputed as time and inventory change. If $v_t=0$, the agent submits no order and uses the canonical offset $\delta_t=0$.\vspace{0.5em}

This approximation therefore fixes the quote between decisions and imposes an explicit submitted-volume cap. The agent does not necessarily expose her whole inventory as the cap can limit execution even when sufficient buying volume reaches her price. Unlike the continuous-time unit-fill model, the simulator permits partial and multi-unit fills, whose total cannot exceed the submitted volume; any unexecuted remainder expires at the next decision time.

\subsubsection{Heuristic Strategy in the Auction}\label{sec:heuristic}

Once the auction opens, an inventory $q_{\tau^\op}\geqslant  0$ remains. We implement the following heuristic policy. Let $\mathcal E=\{i\in\{0,\ldots,n\}:E_{t_i}>0\}$ be the CLOB decisions with positive execution. Each such decision contributes its submitted price once, regardless of the number of shares executed. For $\mathcal E\ne\emptyset$, define
\begingroup
\[
\tilde S=\frac12\left(\frac{1}{|\mathcal E|}\sum_{i\in\mathcal E}S_{t_i}^{\bullet}
+\max_{i\in\mathcal E}S_{t_i}^{\bullet}\right).
\]
\endgroup If no CLOB execution occurred, we set $\tilde S:=H_{t_n}^\cl.$ In $\tilde S$, the mean serves as a proxy for the expected clearing price, and the maximum tilts the limit price upward in the seller's favor.\vspace{0.5em}

The whole remaining inventory $q_{\tau^\op}$ is put on the auction at $\tilde{S}$, with supply function $\tilde{g}_z(p) = \min(q_{\tau^\op}, \min(zq_{\tau^\op}, \beta K_\text{max})(p-\tilde{S})_+)$. In Section \ref{sec:numerics}, we consider $z=10$. This single order is submitted right at the beginning of the auction, and only potentially executed at the clearing time. If $q_t=0$, the benchmark submits no CLOB order; the quote formula is used only for $q_t\geqslant 1$. For both AS and TWAP, the strategic contribution to aggregate net supply is replaced by the capped, one-sided schedule $\tilde{g}_z$. The continuous clearing root is computed using this schedule, which is also evaluated at the bracketing ticks when determining matched volume. The same tick-selection and pro-rata allocation rules then apply. In particular, the benchmark's executed sale cannot exceed its auction-opening inventory.

\subsection{Time-Weighted Average Price Benchmark}\label{section:TWAP}

The second benchmark is deliberately simpler. Given the current inventory $q_t$, the trader submits a deterministic volume $v_t = V_\text{max}\wedge \lfloor q_t\rfloor \wedge  \lceil q_t/(\tau^\op-t)\rceil$, quoted at the best ask price, i.e., $\delta_t = 1$. Such a policy does not guarantee full liquidation during the continuous trading phase, as execution is conditional on order matching. For the auction phase, we adopt exactly the same liquidation policy as in the Avellaneda--Stoikov benchmark, see Section \ref{sec:heuristic}. 

\section{Numerical Simulations}
\label{sec:numerics}

This section evaluates DQN, DDPG, TD3 and SAC against the AS and TWAP references of Section \ref{sec:benchmark}. It reports two families of experiments. Sections \ref{subsec:markets}--\ref{subsec:schemes} specify the markets (synthetic with rough Heston and historical with the midprices of CAT, GOOGL, JPM, MSFT and PG), the protocol and the training schemes of Table \ref{tab:treatment_specification}. Section \ref{subsec:headline} reports results for the headline scheme J-I in all six markets. Section \ref{subsec:contrasts} then compares training design choices through seven matched contrasts between schemes, run on synthetic paths only.\vspace{0.5em}

In every market, a two-hour CLOB phase is followed by a 30-minute closing auction ($\tau^\op=120$ minutes, $\tau^\cl=150$ minutes). This roughly mirrors Xetra, with the CLOB phase standing for the session before the close. The 30-minute auction duration here is, however, a numerical choice. Table \ref{tab:params_generative} collects the simulator, action-grid and preference parameters shared by all experiments; in particular $\alpha=0.01$, $I_0 = 100$ and $\lambda =0.01$.

\subsection{Simulated Markets}
\label{subsec:markets}

\subsubsection{Generative Stochastic Market Model}
\label{sub:generative}

Algorithm \ref{alg:generative_model} generates the exogenous randomness and the accepted exogenous book updates. During continuous trading, market orders arrive on each side as independent Poisson processes with capped Pareto volumes, and the exogenous book has Beta-distributed best-level depth decaying geometrically across levels. At auction opening, the residual exogenous CLOB book is converted into linear carry-over schedules by repeating the slope calibration of Algorithm \ref{alg:hyp_clearing_price} (Section \ref{sec:proj}). During the auction, arrivals and cancellations of market-maker schedules and of market orders are proposed by independent Bernoulli draws at each decision time and accepted only if the resulting exogenous book remains valid, that is, its aggregate slope satisfies $G_t \geqslant D_\mu$ for a numerical lower slope bound $D_\mu >0$ and it admits a nonnegative continuous clearing root as in Proposition \ref{prop:linear}. The within-step chronology of exogenous updates and agent actions is given below.\vspace{0.5em}

\begin{algorithm}[htbp!]
\caption{Exogenous stochastic market generator}
\label{alg:generative_model}
\begin{algorithmic}[1]
\REQUIRE Parameter vector
$(\lambda_0,v_m,\gamma_m,V,\beta_a,\beta_b,V_\infty,
\rho_{\text{lob}},U_1,U_2,L_\infty,D_\mu, B_\infty,p_1,p_2,p_3,p_4)$
\STATE Sample the independent CLOB counting processes
$N^+,N^-\sim\mathrm{PP}(\lambda_0)$ on $[0,\tau^\op]$

\STATE Construct the complete realized grid
$(t_i)_{0\leqslant  i\leqslant  m+1}$ according to
Definition~\ref{def:timing}

\FOR{$k=0,\ldots,n$}
    \FOR{$\zeta\in\{+,-\}$}
        \STATE Sample $V_{t_k}^{\zeta,1}
        \sim
        V_\infty\operatorname{Beta}(\beta_a,\beta_b)$
\FOR{$\ell=1,\ldots,L_{\infty}-1$}
            \STATE Set $V_{t_k}^{\zeta,\ell+1}
            =
            \rho_{\text{lob}} V_{t_k}^{\zeta,\ell}$
        \ENDFOR

        \STATE For $k=n$, use $N_{{\tau^\op}_-}^{\zeta}$ as the upper arrival count in the following loop
        \FOR{$r=N_{t_k}^{\zeta}+1,\ldots,
                   N_{t_{k+1}}^{\zeta}$}
            \STATE Sample $Z_{k,r}^{\zeta}
            \sim
            \operatorname{Pareto}(v_m,\gamma_m)$ and set $\nu_k^{\zeta,r}
            =
            \min\{Z_{k,r}^{\zeta},V\}$
        \ENDFOR
    \ENDFOR

    \STATE Send exogenous CLOB snapshot and arrivals on $(t_k,t_{k+1}]$ to the controlled environment
\ENDFOR

\STATE Receive $O_{{\tau^\op}_-}^{\star}$ from the controlled environment, {recompute the residual-book replacement moments without updating the observed forecast}, initialize the carry-over schedules $\mathcal B_{\tau^\op}^{\star}
=
\{
g_{k,\tau^\op}^{\star}:
k\in\mathcal K_{\star}
\}$ and set $H_{t_{n+1}}^\cl\gets H_{t_n}^\cl$
\STATE Apply the persistent/fallback-schedule rule of Section~\ref{sec:proj} to enforce the initial exogenous slope bound
\FOR{$k=n+1,\ldots,m$}
    \STATE Sample proposed auction updates $B_{t_k}\sim\mathcal B(p_1)$, $D_{t_k}\sim\mathcal B(p_2)$, $J_{t_k}^{\pm}\sim\mathcal B(p_3)$, $G_{t_k}^{\pm}\sim\mathcal B(p_4)$

    \IF{$B_{t_k}=1$}
        \STATE Sample $K_{t_k}^r\sim\mathcal U([U_1,U_2])$, $S_{t_k}^r
        \sim
        S_{\tau^\op}^{\mathrm{mid}}
        +
        \alpha\mathcal U(\llbracket -B_\infty,B_\infty\rrbracket)$
        and add the new exogenous schedule only if the resulting exogenous book remains valid
    \ENDIF

    \IF{$D_{t_k}=1$ and a nonpersistent exogenous schedule is present}
        \STATE Propose canceling one uniformly selected nonpersistent live exogenous schedule and accept the cancellation only if the resulting exogenous book remains valid
    \ENDIF

    \FOR{$\zeta\in\{+,-\}$}
        \IF{$J_{t_k}^{\zeta}=1$}
            \STATE Sample $Z_{t_k,\zeta}
            \sim
            \operatorname{Pareto}(v_m,\gamma_m)$ and set $\nu_{t_k}^{\zeta}
            =
            \min\{Z_{t_k,\zeta},V\}$
            \STATE Add the corresponding market order if the resulting book is valid
        \ENDIF

        \IF{$G_{t_k}^{\zeta}=1$ and an active market order is present
        on side $\zeta$}
            \STATE Cancel a uniformly sampled order if the resulting book is valid
        \ENDIF
    \ENDFOR
\ENDFOR
\end{algorithmic}
\end{algorithm}
In the CLOB loop of Algorithm~\ref{alg:generative_model}, $\nu_k^{\zeta,r}$ denotes the volume of order $r$ arriving during $(t_k,t_{k+1}]$. For $k<n$, it supplies $\nu_{t_{k+1}}^{\zeta,r}$ in the execution equations; for $k=n$, the CLOB arrivals exclude the auction-opening update and supply $\nu_{{\tau^\op}_-}^{\zeta,r}$.
At each auction decision time, the accepted exogenous update returned by Algorithm \ref{alg:generative_model} is incorporated into the current configuration before the agent acts via the following chronology.
\begin{enumerate}
    \item If $k=n+1$, set
    $H_{t_{n+1}}^\cl=H_{t_n}^\cl$. If $k\geqslant n+2$, compute
    $H_{t_k}^\cl$ from the complete post-action book stored at
    $t_{k-1}$.
    \item Sample the exogenous proposals at $t_k$ and accept each
    proposal only if the resulting exogenous book remains valid.
    \item Construct the pre-action configuration $X_{t_k}$ and reveal
    $\varphi(X_{t_k})$ to the agent.
    \item If $c_{t_k}=1$, remove every still-active agent schedule
    submitted strictly before $t_k$.
    \item Add the current schedule $p\longmapsto
    K_{t_k}^a
    \left(
    p-S_{\tau^\op}^{\mathrm{mid}}-\alpha b_{t_k}^a
    \right).$
    \item Store the resulting post-action book. If $k=m$, freeze this
    book, process no event on $(t_m,\tau^\cl]$, and apply the terminal
    clearing and allocation rule.
\end{enumerate}

{The emulator parameters are
\[
\bigl(\lambda_0,v_m,\gamma_m,V,\beta_a,\beta_b,V_\infty,
\rho_{\mathrm{lob}},U_1,U_2,L_\infty,D_\mu,B_\infty,p_1,p_2,p_3,p_4\bigr).
\]}
\begin{itemize}
    \item The parameter $\lambda_0$ is the arrival intensity, measured in arrivals per side per minute, of the independent Poisson processes $N^+$ and $N^-$ during the continuous phase. They are sampled in the first step of the algorithm.
    \item $(v_m,\gamma_m)$ are the parameters of the Pareto distribution modeling the volumes of the market orders throughout the trading session. This reproduces the well-known heavy-tailed behavior of the density of market order size, see for example \cite{gopikrishnan2000statistical,gabaix2009power,bouchaud2001power}. The market model allows orders of maximum size $V$: proposed order volumes above $V$ are capped at $V$, creating a point mass there.
    \item $(\beta_a,\beta_b)$ parameterize the Beta distribution of the volume at the first price level on each side of the limit-order book. {The exogenous best-level depth law $V_\infty\operatorname{Beta}(\beta_a,\beta_b)$ is a modeling choice.}  Here $V_\infty$ is a rescaling parameter as the Beta distribution has support $[0,1]$. Further volumes in the limit order book decay geometrically with parameter $\rho_{\text{lob}} \in (0,1]$ and are truncated via the maximum order book depth $L_\infty$, that is $V_t^{\zeta,j}=0$ for $j > L_\infty$.
    \item $(U_1,U_2)$ are the lower and upper bounds of the uniform distribution from which exogenous market makers sample their auction supply slopes.
    \item $B_\infty$ {represents} the bound, in ticks, on new exogenous auction reference-price deviations from the frozen opening midprice. {Exogenous market makers during the auction have reference prices sampled as} $S_t^i \sim S_{\tau^\op}^\mathrm{mid} + \alpha \mathcal{U}(\llbracket -B_\infty,B_\infty \rrbracket)$ because the auction opening time is the last time market participants see the midprice. {A similar auction reference-price specification has been introduced} in \cite{derchu2023equilibria}.
    \item {At each one-minute auction decision, independent Bernoulli proposals govern market-maker arrival ($p_1$), cancellation of a nonpersistent market-maker schedule ($p_2$), and, independently on each side, market-order arrival ($p_3$) and cancellation of an active market order ($p_4$). Each proposed update is accepted only if the resulting exogenous book is valid.}
\end{itemize}

The proposal variables and their marks are sampled independently. The accepted book updates are state-dependent because of the validity screen. We assume that the limit order book is refreshed from one time step to the next one.\vspace{0.5em}

\begingroup
\begin{table}[htbp!]

\centering

\setlength{\tabcolsep}{4pt}
\begin{tabular}{lcp{0.59\linewidth}}
\toprule
\textbf{Symbol} & \textbf{Value} & \textbf{Meaning in model units} \\
\midrule
$[t]$ & minutes & Unit of the simulator clock \\
$\tau^{\mathrm{op}}$ & $120$ & Auction opening time \\
$\tau^{\mathrm{cl}}$ & $150$ & Clearing time; 30 auction decisions \\
$S_0^{\mathrm{mid}}$ & $100$ & Rebased initial midprice \\
$I_0$ & $100$ & Initial model inventory \\
$\lambda_0$ & $0.5$/min/side & CLOB Poisson intensity per side \\
$v_m$ & $2$ & Pareto scale (model inventory units) \\
$\gamma_m$ & $2.5$ & Pareto shape \\
$V_\infty$ & $2$ & Top-depth Beta scaling \\
$\beta_a,\beta_b$ & $2,5$ & Top-depth Beta shapes \\
$\rho_{\mathrm{lob}}$ & $0.96$ & Geometric depth decay \\
$V$ & $30$ & Exogenous market-order volume cap \\
$V_{\max}$ & $30$ & Strategic CLOB submitted-volume cap \\
$L_\infty$ & $200$ & Exogenous LOB depth in levels \\
$L_{\max}$ & $12$ & Strategic CLOB quote-offset bound in ticks \\
$B_\infty$ & $150$ & Reference-price deviation bound for strategic schedules and new exogenous arrivals, in ticks from frozen mid \\
$B_{\max}$ & $10$ & Local strategic offset half-width around the selected anchor, in ticks \\
$D_\mu$ & $0.1$ & Lower bound on valid exogenous aggregate slope \\
$U_1,U_2$ & $1,20$ & Exogenous per-schedule slope support endpoints \\
$p_1$ & $1$ & Market-maker arrival proposal probability per auction minute \\
$p_2$ & $0$ & Exogenous market-maker cancellation proposal probability \\
$p_3$ & $0.3$ & Market-taker arrival proposal probability per side per minute \\
$p_4$ & $0.05$ & Active market-order cancellation proposal probability per side per minute \\
$\lambda$ & $0.01$ & Terminal quadratic inventory-risk coefficient \\
$q$ & $0$ & Purchase attenuation only; other preferences remain active \\
$k^\star$ & $10{,}000$ & Active CLOB preference clipping scale \\
$\omega_A$ & $10^{-4}$ & Active fictive auction-credit weight \\
$d$ & $0.001$ & Stylized withdrawal-cost increment, {$d_{t_j}=(j-n-1)d$} \\
$\alpha$ & $0.01$ & Tick in rebased model-price units \\
$\beta$ & $0.25$ & Strategic slope quantum (inventory units per price unit) \\
$K_{\max}$ & $32$ & Maximum integer slope index, with zero included \\
$\beta K_{\max}$ & $8$ & Maximum physical slope of one strategic schedule in model units \\
$z$ & 10 & Heuristically chosen inverse-price parameter\\
$H_0$ & $100$ & Initial raw-recursion clearing signal \\
$\eta_H$ & $0.95$ & Raw clearing-price smoothing weight in Algorithm \ref{alg:hyp_clearing_price} \\
\bottomrule
\end{tabular}
\caption{Shared simulator, action-grid and preference parameters for the synthetic and historical headlines.}
\label{tab:params_generative}
\end{table}
\endgroup

\subsubsection{Synthetic and Historical Midprice Paths}
\label{sub:rough}

In the synthetic market, $S_t$ denotes the unrounded price process underlying the midprice and follows a rough Heston model, motivated by the so-called rough volatility of financial assets \cite{gatheral2022volatility,bäuerle2020portfolio,abi2019multifactor}. Let $\rho\in[-1,1]$, $S_0=100$, and $V_0,\theta,\varsigma,\nu>0$, with $0<H<1/2$ in the rough-volatility regime. Recalling \cite{gatheral2022volatility,richard2023discrete}, the rough Heston model is given by
\[\left\{\begin{array}{l}
    S_t = S_0 + \int_0^t S_s \sqrt{V_s}\mathrm{d}\big( \rho B_s + \sqrt{1-\rho^2}B_s^\perp\big), \\
    V_t = V_0 + \int_0^t K_H(t-s)\big ((\theta-\varsigma V_s)\mathrm{d}s + \nu \sqrt{V_s}\mathrm{d}B_s \big),
\end{array}\right.\]
where $(B,B^\perp)$ are two independent Brownian motions and $K_H(t) = t^{H-1/2}/\Gamma(H+1/2)$ for $t>0$, with the illustrative parameters $H=0.1$, $\rho=-0.7$, $V_0=0.02$, $\theta=0.02$, $\varsigma=0.3$ and $\nu=0.3$, adapted from \cite{richard2023discrete}.\vspace{0.5em}

Set $Y_0=\log S_0$. We use the Euler-type scheme of \cite{richard2023discrete} on a refined price-simulation grid. For each episode, we first construct the realized CLOB decision grid $0=t_0<\cdots<t_n={\tau^\op}-1<t_{n+1}=\tau^\op$. We subdivide each interval $[t_i,t_{i+1}]$ into the smallest power-of-two number of equal subintervals whose lengths do not exceed $\delta_{\mathrm{int}}=0.25$ minutes (15 seconds). Their endpoints form the refined grid $0=u_0<\cdots<u_N=\tau^\op$. This grid contains every CLOB decision time and the auction opening, and satisfies $\max_{0\leqslant i<N}(u_{i+1}-u_i)\leqslant\delta_{\mathrm{int}}$. This construction gives nested grids as the step bound is tightened.

\begin{remark}
    We use a maximum internal step of $\delta_{\mathrm{int}}=0.25$ minutes to increase temporal resolution while keeping repeated policy training computationally feasible. While coupled refinement checks indicate residual mesh sensitivity, finite-grid accuracy has not been established. Accordingly, the synthetic results are conditional on this discretized simulator.
\end{remark}

 We measure model time in trading years: for $s^\star=252\cdot6.5\cdot60$, set $\bar u_i=u_i/s^\star$ and $\Delta\bar u_{i+1}=\bar u_{i+1}-\bar u_i$. Using Brownian drivers independent of the order-flow processes, we simulate the rough Heston approximation conditional on the complete refined grid. Write $\Delta B_{i+1}=B_{\bar u_{i+1}}-B_{\bar u_i}$ and $\Delta B_{i+1}^{\perp}=B_{\bar u_{i+1}}^{\perp}-B_{\bar u_i}^{\perp}$ for $i\in\llbracket 0,N-1\rrbracket$. Conditional on the grid, these increments are independent across intervals and drivers, centered Gaussian, and have variance $\Delta\bar u_{i+1}$. For $k\in\llbracket 1,N\rrbracket$, the log-price and variance approximations at $u_k$ are
\[
\left\{\begin{aligned}
Y_k &= Y_0+\sum_{i=0}^{k-1}\left[
-\frac12(V_i)_+\Delta\bar u_{i+1}
+\sqrt{(V_i)_+}\left(\rho\Delta B_{i+1}
+\sqrt{1-\rho^2}\Delta B_{i+1}^{\perp}\right)\right],\\
V_k &= V_0+\sum_{i=0}^{k-1}K_H(\bar u_k-\bar u_i)\left[
(\theta-\varsigma(V_i)_+)\Delta\bar u_{i+1}
+\nu\sqrt{(V_i)_+}\Delta B_{i+1}\right].
\end{aligned}\right.
\]

Let $\kappa(i)$ satisfy $u_{\kappa(i)}=t_i$. The trading environment observes the price only at the original decision times and auction opening, with $S_{t_i}^{\mathrm{mid}}=\alpha\left\lfloor \exp(Y_{\kappa(i)})/\alpha+1/2 \right\rfloor$. The midprice is simulated through $u_N=\tau^\op$ and then frozen at $S_{\tau^\op}^{\mathrm{mid}}$ throughout the auction. The opening forecast inherits the last pre-action CLOB forecast.\vspace{0.5em}

The AS reference benchmark assumes $\mathrm{d}S_t = \sigma \mathrm{d}B_t$, rather than rough Heston. As $\sigma$ drops out of the risk-neutral quotes of Proposition \ref{prop:benchm}, we omit its calibration. Following the proportional power-law argument in \cite{avellaneda2008high}, we adopt the heuristic parameterization $A=\lambda_0/\gamma_m$ and $k=\gamma_m\hat K$. Note that $A$ is per minute. Here, $\hat{K}$ is the proportionality constant in the empirical relationship $\hat{K}\Delta p = \ln(Q)$, where $\Delta p$ is the move in price when a market order of size $Q$ arrives. We fit $\hat{K}$ by least squares on 5,000 draws of the book and of market orders $Q = \min(Z,V)$ for $Z \sim \mathrm{Pareto}(v_m,\gamma_m)$. Draws that empty the ask book or do not move its best price are discarded before fitting $\hat K\Delta p=\ln Q$ through the origin. In the original argument, the exponential execution intensity is motivated by a market-order size density specified up to proportionality and a logarithmic price-impact relationship. We use this construction as a heuristic calibration. In particular, the normalized, capped Pareto distribution and random book depth do not imply the same exponential fill law. Moreover, the simulator permits partial and multi-unit executions over decision intervals, whereas the quote derivation assumes unit transactions.\vspace{0.5em}

In the five historical markets, the observed midprice path of CAT, GOOGL, JPM, MSFT or PG replaces the rough Heston path, while order flow and the auction remain simulated as in Section \ref{sub:generative}. Historical training draws whole sessions from all five tickers over August 3--14, 2026. Validation and testing are conducted separately for each ticker over August 17--21 and August 24--28, respectively. All historical midprices are projected onto the tick grid after rebasing, using nearest-tick rounding with half-ticks rounded upward.

\subsection{Experimental Protocol}
\label{subsec:experimental_protocol}
We train each learner with ten master seeds under each training scheme of Section \ref{subsec:schemes}. Forecasts are fitted using training data only, as described in Section \ref{sec:proj}; input transformations and reference values follow Section \ref{subsec:conditioning}. Calibration, training, validation and testing use separate simulation streams. Training is limited to 800 episodes per run. After the structured-exploration episodes and once each phase buffer holds its replay minimum (Table \ref{tab:hyperparams}), we evaluate policies on a fixed set of 128 validation episodes every 50 completed training episodes. Among eligible evaluations, we retain the joint CLOB--auction checkpoint with the lowest mean inventory-penalized shortfall, $\mathrm{IS}_\lambda$. We also stop training early after four consecutive eligible evaluations without a strict improvement in validation $\mathrm{IS}_\lambda$. Each selected policy is evaluated deterministically on 100 test episodes. Learned policies, untrained models, and the AS and TWAP references face the same episode and environment-seed pairs. We preserve this pairing across learners and training schemes within each master seed and market. In the historical markets, the 100 test episodes reuse the five test dates with different simulated order-flow and auction draws.\vspace{0.5em}

Unless stated otherwise, a reported value is the mean over the ten master seeds of each seed's mean over its 100 test episodes, and historical aggregates are equal-ticker means over the five tickers. Figure \ref{fig:economic_shortfall} reports $-\mathrm{IS}_\lambda$, so higher values are better; Table \ref{tab:treatments} and Figures \ref{fig:forecast_information} and \ref{fig:learning_guidance} report its difference between the first- and second-named scheme, so positive values favor the first.

\subsection{Training Schemes}
\label{subsec:schemes}
Table \ref{tab:treatment_specification} lists the eight training schemes. The prefix gives the training reward: J-schemes optimize the weighted objective $J$ and E-schemes the economic objective, both with learning conditioning, while Cash-I optimizes the economic objective without any conditioning, which is just raw cash flow. The suffix gives the auction anchor (I: indicative; F: frozen midprice) or, for E-S and E-0, the single change relative to E-F (sparse auction credit; no explicit $H^\cl$ input). No-A removes auction access: residual inventory is marked and penalized at auction opening rather than forcibly executed, and shortfall is computed with zero auction fill and the opening-time residual position. J-I is the headline scheme and the only one trained on historical midprices; all results in Section \ref{subsec:headline} refer to it, and the other schemes enter only the synthetic contrasts of Section \ref{subsec:contrasts}.

\begin{table}[H]
\centering
\setlength{\tabcolsep}{4pt}
\begin{tabular}{@{}ccccccc@{}}
\toprule
\textbf{Scheme} & \textbf{Objective} & \textbf{Conditioning} & \textbf{Explicit $H^\cl$} & \textbf{Anchor} & \textbf{Auction} & \textbf{Table \ref{tab:treatments}} \\
\midrule
J-I & $J$ & Common + dense & On & Indicative & On & C3, C6, C7 \\
J-F & $J$ & Common + dense & On & Frozen mid & On & C3, C4 \\
E-F & Economic & Common + dense & On & Frozen mid & On & C1, C2, C4 \\
E-S & Economic & Common + sparse & On & Frozen mid & On & C1 \\
E-0 & Economic & Common + dense & Off & Frozen mid & On & C2, C5 \\
No-A & Economic & Common, CLOB only & Off & -- & Off & C5 \\
Cash-I & Economic & None & On & Indicative & On & C6 \\
E-I & Economic & Common + dense & On & Indicative & On & C7 \\
\bottomrule
\end{tabular}
\caption{Training schemes}
\label{tab:treatment_specification}
\end{table}

\subsection{Headline Results}
\label{subsec:headline}
All results in this section use the headline scheme J-I.

\subsubsection{Economic Performance}
The learned policies improve the trade-off between execution and terminal inventory relative to the AS and TWAP references (Table \ref{tab:economic_performance}). Table \ref{tab:benchmark_supplement} reports the paired comparisons of each learner against the AS and TWAP references and the market-specific ordinary-shortfall levels behind the equal-ticker historical means of Table \ref{tab:economic_performance}. Every learner has lower inventory-penalized shortfall than both references in all six markets. No learner dominates the others. In the synthetic market, SAC has the lowest mean, but its paired comparisons with the other learners remain uncertain. Figure \ref{fig:headline_learning} shows how validation performance varies across seeds, which checkpoints are selected, and how long the runs continue.

\begin{table}[H]
\centering
\setlength{\tabcolsep}{7pt}
\begin{tabular}{@{}lrrrr@{}}
\toprule
& \multicolumn{2}{c}{\textbf{Synthetic}} & \multicolumn{2}{c}{\textbf{Historical, equal-ticker mean}} \\
\cmidrule(lr){2-3}\cmidrule(l){4-5}
\textbf{Policy} & {\boldmath$\mathrm{IS}$} & \textbf{\boldmath$\mathrm{IS}_\lambda$ [95\% interval]} & {\boldmath$\mathrm{IS}$} & \textbf{\boldmath$\mathrm{IS}_\lambda$ [95\% interval]} \\
\midrule
DQN & $-5.11$ & $-4.19\;[-5.85,-2.76]$ & $-7.12$ & $-5.97\;[-6.42,-5.52]$ \\
DDPG & $-5.08$ & $-4.03\;[-5.97,-2.32]$ & $-7.74$ & $-6.68\;[-7.04,-6.31]$ \\
TD3 & $-5.28$ & $-4.40\;[-6.25,-2.75]$ & $-7.27$ & $-6.39\;[-6.70,-6.05]$ \\
SAC & $-5.39$ & $-4.54\;[-6.35,-3.15]$ & $-7.35$ & $-6.44\;[-6.86,-5.99]$ \\
\midrule
AS & $-6.45$ & $-1.41\;[-2.98,-0.06]$ & $-8.44$ & $-3.39\;[-4.06,-2.71]$ \\
TWAP & $-0.36$ & $5.61\;[4.10,7.01]$ & $-3.46$ & $2.53\;[2.27,2.79]$ \\
\bottomrule
\end{tabular}
\caption{Test shortfalls of the headline scheme J-I and the references, in initial-notional basis points. Historical: equal-ticker mean over the five tickers.}
\label{tab:economic_performance}
\end{table}

\begin{table}[H]
\centering\small
\setlength{\tabcolsep}{5pt}
\renewcommand{\arraystretch}{1}
\begin{tabular}{@{}llrrr@{}}
\toprule
\textbf{Market} & \textbf{Policy} & \textbf{Ordinary shortfall} & \multicolumn{2}{c}{\textbf{Paired penalized-shortfall reduction}} \\
\cmidrule(l){4-5}
 & & $\mathrm{IS}$ [95\% CI] & versus AS [95\% CI] & versus TWAP [95\% CI] \\
\midrule
Synthetic & DQN & $-5.11\;[-6.76,-3.76]$ & $2.79\;[2.25,3.42]$ & $9.81\;[8.92,10.59]$ \\
 & DDPG & $-5.08\;[-7.04,-3.27]$ & $2.63\;[1.70,3.50]$ & $9.65\;[8.72,10.53]$ \\
 & TD3 & $-5.28\;[-7.15,-3.58]$ & $3.00\;[2.14,3.92]$ & $10.02\;[9.11,10.81]$ \\
 & SAC & $-5.39\;[-7.30,-3.95]$ & $3.13\;[2.69,3.56]$ & $10.15\;[9.32,11.03]$ \\
 & AS & $-6.45\;[-8.27,-4.98]$ & --- & --- \\
\addlinespace[5pt]
CAT & DQN & $-16.18\;[-16.98,-15.30]$ & $2.31\;[1.80,2.86]$ & $8.27\;[7.71,8.78]$ \\
 & DDPG & $-17.60\;[-18.60,-16.51]$ & $4.11\;[3.33,4.90]$ & $10.07\;[9.65,10.54]$ \\
 & TD3 & $-15.76\;[-16.40,-15.13]$ & $2.57\;[1.58,3.44]$ & $8.53\;[7.79,9.23]$ \\
 & SAC & $-16.32\;[-17.14,-15.41]$ & $3.04\;[2.39,3.72]$ & $9.00\;[8.38,9.62]$ \\
 & AS & $-17.61\;[-18.39,-16.69]$ & --- & --- \\
\addlinespace[5pt]
GOOGL & DQN & $-4.10\;[-4.82,-3.20]$ & $3.42\;[2.84,4.01]$ & $10.26\;[9.90,10.64]$ \\
 & DDPG & $-4.10\;[-4.85,-3.25]$ & $3.53\;[2.87,4.23]$ & $10.37\;[10.10,10.63]$ \\
 & TD3 & $-4.39\;[-5.18,-3.49]$ & $3.98\;[3.10,4.82]$ & $10.82\;[10.26,11.37]$ \\
 & SAC & $-4.23\;[-4.99,-3.30]$ & $3.87\;[2.77,4.95]$ & $10.71\;[10.13,11.28]$ \\
 & AS & $-4.51\;[-4.96,-3.96]$ & --- & --- \\
\addlinespace[5pt]
JPM & DQN & $-5.64\;[-6.54,-4.71]$ & $1.58\;[0.80,2.38]$ & $6.99\;[6.27,7.68]$ \\
 & DDPG & $-6.95\;[-8.05,-5.78]$ & $2.84\;[2.03,3.68]$ & $8.24\;[7.58,8.82]$ \\
 & TD3 & $-6.47\;[-7.49,-5.47]$ & $2.42\;[1.37,3.36]$ & $7.83\;[7.02,8.55]$ \\
 & SAC & $-6.64\;[-7.54,-5.72]$ & $2.57\;[1.82,3.34]$ & $7.98\;[7.39,8.49]$ \\
 & AS & $-8.26\;[-8.94,-7.53]$ & --- & --- \\
\addlinespace[5pt]
MSFT & DQN & $0.87\;[-0.52,2.34]$ & $3.30\;[2.69,3.76]$ & $9.11\;[8.29,9.97]$ \\
 & DDPG & $0.92\;[-0.26,2.02]$ & $3.17\;[2.64,3.70]$ & $8.99\;[8.23,9.76]$ \\
 & TD3 & $1.05\;[-0.42,2.65]$ & $3.36\;[2.16,4.40]$ & $9.18\;[8.23,10.13]$ \\
 & SAC & $1.37\;[-0.24,3.06]$ & $3.03\;[1.50,4.42]$ & $8.85\;[7.82,9.89]$ \\
 & AS & $0.61\;[-0.60,1.76]$ & --- & --- \\
\addlinespace[5pt]
PG & DQN & $-10.54\;[-11.41,-9.55]$ & $2.28\;[1.79,2.77]$ & $7.86\;[7.59,8.11]$ \\
 & DDPG & $-10.95\;[-11.89,-9.92]$ & $2.78\;[1.99,3.49]$ & $8.37\;[8.01,8.71]$ \\
 & TD3 & $-10.76\;[-11.58,-9.93]$ & $2.67\;[1.74,3.48]$ & $8.26\;[7.74,8.72]$ \\
 & SAC & $-10.92\;[-11.74,-9.96]$ & $2.73\;[1.73,3.72]$ & $8.31\;[7.78,8.87]$ \\
 & AS & $-12.42\;[-13.26,-11.51]$ & --- & --- \\
\bottomrule
\end{tabular}
\caption{Market-specific shortfalls in basis points}
\label{tab:benchmark_supplement}
\end{table}

\begin{figure}[H]

\centering
\includegraphics[width=\linewidth]{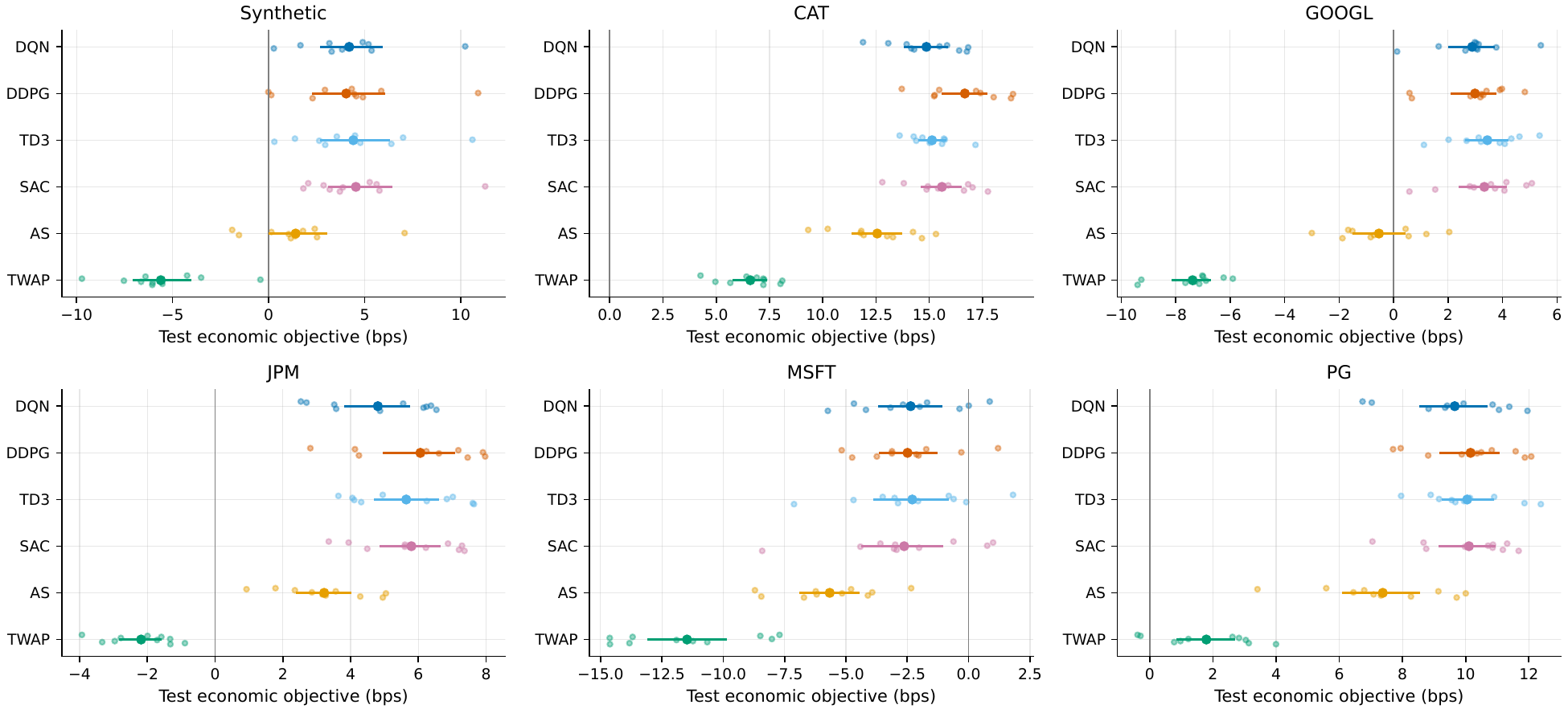}
\caption{Test $-\mathrm{IS}_\lambda$ of the headline J-I policies and the references by market}
\label{fig:economic_shortfall}
\end{figure}

\begin{figure}[H]

\centering
\includegraphics[width=\linewidth]{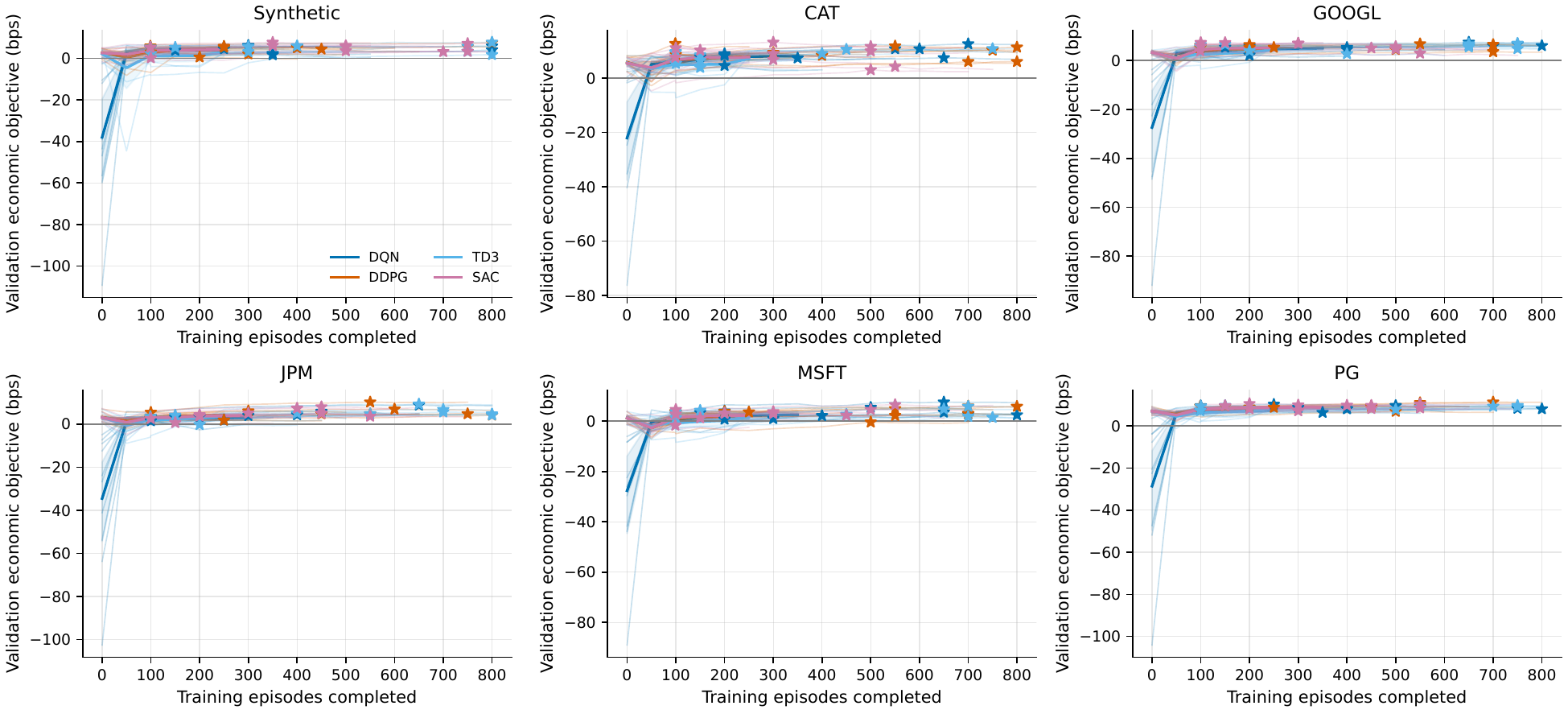}
\caption{Validation economic score $-\mathrm{IS}_\lambda$}
\label{fig:headline_learning}
\end{figure}

\subsubsection{Auction Contribution and Inventory}
\label{subsec:auction_contribution}
How much does the auction itself contribute? Holding a learner's CLOB trajectory fixed, its auction contribution is the reduction in penalized shortfall relative to submitting no auction orders:
\[
Z_{\tau^\cl}\bigl(S_{\tau^\cl}^\cl-S_{\tau^\op}^{\mathrm{mid}}\bigr)-\sum_{i=n+1}^{m}d_{t_i}c_{t_i}
+\lambda\bigl(I_{\tau^\op}^2-I_{\tau^\cl}^2\bigr).
\]
Figure \ref{fig:auction_value_inventory} (top) shows positive mean contributions in the synthetic market and in the equal-ticker historical mean. The continuous learners obtain larger mean contributions than DQN in both summaries. In the historical mean, most of the contribution comes from reducing inventory risk rather than improving execution prices. These comparisons follow each learner's own CLOB policy. The policies sell most inventory before the call. Between auction opening and clearing, synthetic DQN, TD3 and SAC increase mean absolute inventory while reducing its mean quadratic penalty.

\begin{figure}[htbp!]

\centering
\includegraphics[width=\linewidth]{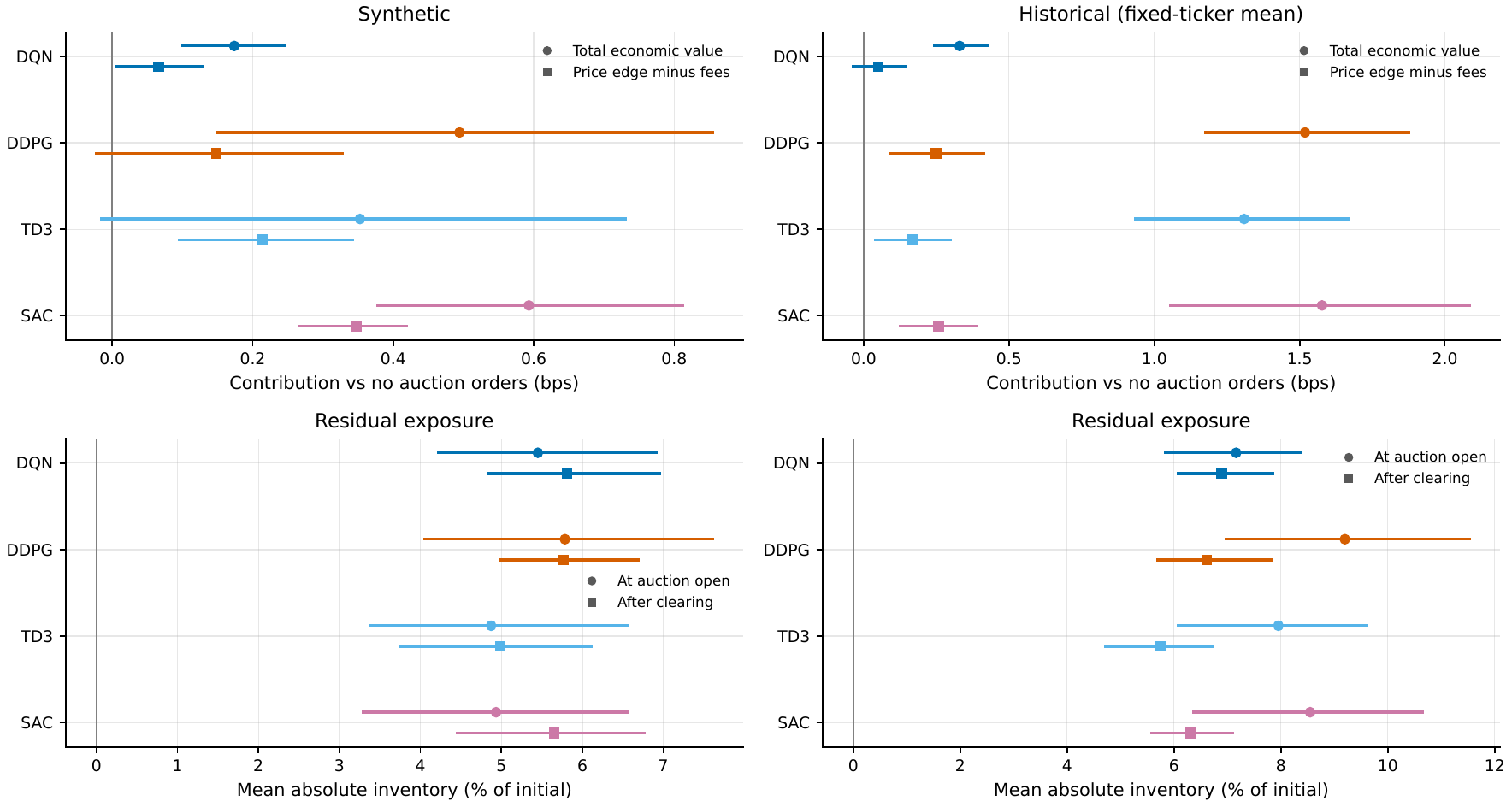}
\caption{Auction contribution and mean absolute inventory at auction opening and after clearing for the headline J-I policies, each on its own CLOB trajectory. In the top panels, the gap between the two markers is the inventory-penalty term $\lambda(I_{\tau^\op}^2-I_{\tau^\cl}^2)$ of the contribution.}
\label{fig:auction_value_inventory}
\end{figure}

\subsubsection{Forecast Accuracy}

\begin{table}[htbp!]

\centering
\setlength{\tabcolsep}{4pt}
\renewcommand{\arraystretch}{1.04}
\begin{tabular}{llrrrrr}
\toprule
& & \multicolumn{3}{c}{\textbf{Mean absolute error (MAE)}} & \multicolumn{2}{c}{\textbf{RMSE}} \\
\cmidrule(lr){3-5}\cmidrule(lr){6-7}
\textbf{Market} & \textbf{Phase} & \textbf{\boldmath$H^\cl$} & \textbf{Midprice} & \textbf{Reduction [95\% CI]} & \textbf{\boldmath$H^\cl$} & \textbf{Midprice} \\
\midrule
Synthetic & CLOB & 0.1726 & 0.1837 & 0.0111 [0.0097, 0.0126] & 0.2775 & 0.2831 \\
 & Auction & 0.0504 & 0.1471 & 0.0967 [0.0884, 0.1050] & 0.0908 & 0.1964 \\
\addlinespace[4pt]
CAT & CLOB & 0.1127 & 0.1314 & 0.0187 [0.0177, 0.0197] & 0.1675 & 0.1811 \\
 & Auction & 0.0331 & 0.1200 & 0.0868 [0.0813, 0.0921] & 0.0505 & 0.1598 \\
\addlinespace[4pt]
GOOGL & CLOB & 0.0806 & 0.0989 & 0.0183 [0.0177, 0.0188] & 0.1132 & 0.1313 \\
 & Auction & 0.0267 & 0.1025 & 0.0758 [0.0712, 0.0814] & 0.0454 & 0.1308 \\
\addlinespace[4pt]
JPM & CLOB & 0.1169 & 0.1348 & 0.0179 [0.0165, 0.0194] & 0.1586 & 0.1720 \\
 & Auction & 0.0345 & 0.1470 & 0.1124 [0.1082, 0.1164] & 0.0532 & 0.1702 \\
\addlinespace[4pt]
MSFT & CLOB & 0.1273 & 0.1485 & 0.0212 [0.0203, 0.0221] & 0.1859 & 0.2077 \\
 & Auction & 0.0413 & 0.1579 & 0.1166 [0.1084, 0.1243] & 0.0680 & 0.1869 \\
\addlinespace[4pt]
PG & CLOB & 0.0752 & 0.0863 & 0.0111 [0.0106, 0.0116] & 0.1096 & 0.1176 \\
 & Auction & 0.0229 & 0.0963 & 0.0734 [0.0683, 0.0783] & 0.0353 & 0.1163 \\
\bottomrule
\end{tabular}
\caption{Accuracy of $H_t^\cl$ and of the contemporaneous midprice as predictors of the clearing price $S_{\tau^\cl}^\cl$ along test paths of the headline J-I policies, by market and trading phase, in model-price units}
\label{tab:forecast_accuracy}
\end{table}

Table \ref{tab:forecast_accuracy} displays the forecast errors. As a predictor of the clearing price, $H_t^\cl$ has lower mean absolute error (MAE) and root mean squared error (RMSE) than the contemporaneous midprice in both phases of every market.

\subsection{Matched Synthetic Contrasts}
\label{subsec:contrasts}
Table \ref{tab:treatments} reports seven matched contrasts on synthetic paths; Figures \ref{fig:forecast_information} and \ref{fig:learning_guidance} show their seed-level distributions.

\begin{table}[H]
\centering
\setlength{\tabcolsep}{4pt}
\renewcommand{\arraystretch}{1.12}
\begin{tabular}{@{}lcccc@{}}
\toprule
\textbf{Contrast (first versus second)} & \textbf{DQN} & \textbf{DDPG} & \textbf{TD3} & \textbf{SAC} \\
\midrule
\shortstack[l]{\textbf{C1}: Dense versus sparse credit\\E-F versus E-S}
& \shortstack{$0.73$\\\scriptsize $[0.18,1.19]$}
& \shortstack{$2.08$\\\scriptsize $[0.95,3.23]$}
& \shortstack{$1.78$\\\scriptsize $[0.82,2.70]$}
& \shortstack{$0.74$\\\scriptsize $[0.24,1.26]$} \\
\addlinespace[4pt]
\shortstack[l]{\textbf{C2}: Explicit $H^\cl$ input\\E-F versus E-0}
& \shortstack{$-0.44$\\\scriptsize $[-0.78,-0.09]$}
& \shortstack{$-0.32$\\\scriptsize $[-0.94,0.11]$}
& \shortstack{$0.03$\\\scriptsize $[-0.52,0.57]$}
& \shortstack{$0.44$\\\scriptsize $[0.00,0.95]$} \\
\addlinespace[4pt]
\shortstack[l]{\textbf{C3}: Indicative versus frozen-mid anchor\\J-I versus J-F}
& \shortstack{$-0.44$\\\scriptsize $[-0.90,0.07]$}
& \shortstack{$-0.37$\\\scriptsize $[-0.72,-0.06]$}
& \shortstack{$-0.42$\\\scriptsize $[-0.67,-0.07]$}
& \shortstack{$-0.92$\\\scriptsize $[-1.40,-0.48]$} \\
\addlinespace[4pt]
\shortstack[l]{\textbf{C4}: Preferences, frozen-mid anchor\\J-F versus E-F}
& \shortstack{$0.02$\\\scriptsize $[-0.32,0.35]$}
& \shortstack{$-0.20$\\\scriptsize $[-0.46,0.09]$}
& \shortstack{$-0.13$\\\scriptsize $[-0.34,0.06]$}
& \shortstack{$-0.02$\\\scriptsize $[-0.24,0.24]$} \\
\addlinespace[4pt]
\shortstack[l]{\textbf{C5}: Retrained auction access\\E-0 versus No-A}
& \shortstack{$0.57$\\\scriptsize $[0.12,1.08]$}
& \shortstack{$1.05$\\\scriptsize $[0.75,1.26]$}
& \shortstack{$0.62$\\\scriptsize $[0.30,0.93]$}
& \shortstack{$0.77$\\\scriptsize $[0.14,1.41]$} \\
\addlinespace[4pt]
\shortstack[l]{\textbf{C6}: Complete scheme versus raw cash flow\\J-I versus Cash-I}
& \shortstack{$3.53$\\\scriptsize $[1.10,6.64]$}
& \shortstack{$3.65$\\\scriptsize $[2.84,4.57]$}
& \shortstack{$4.83$\\\scriptsize $[4.03,5.70]$}
& \shortstack{$5.16$\\\scriptsize $[4.08,6.27]$} \\
\addlinespace[4pt]
\shortstack[l]{\textbf{C7}: Preferences, indicative anchor\\J-I versus E-I}
& \shortstack{$-0.01$\\\scriptsize $[-0.60,0.53]$}
& \shortstack{$0.04$\\\scriptsize $[-0.60,0.57]$}
& \shortstack{$-0.12$\\\scriptsize $[-0.31,0.06]$}
& \shortstack{$-0.16$\\\scriptsize $[-0.30,-0.05]$} \\
\bottomrule
\end{tabular}
\caption{Matched synthetic contrasts in basis points (schemes in Table \ref{tab:treatment_specification})}
\label{tab:treatments}
\end{table}

Because auction orders settle only at the close, a learner receives little immediate economic feedback about its decisions. Against sparse credit, dense auction credit reduces test $\mathrm{IS}_\lambda$ for all four learners. The complete headline scheme also outperforms raw cash-flow training (Figure \ref{fig:learning_cash}). The matched preference contrasts show no clear marginal benefit from the preference terms of $J$. Regarding the use of the forecast, Table \ref{tab:forecast_accuracy} shows that $H_t^\cl$ outperforms the midprice. The contrast C2 (E-F versus E-0) jointly changes access to the calibrated CLOB forecast and the inherited or lagged auction indication. Including the explicit feature worsens DQN; the DDPG, TD3 and SAC intervals include zero. Under $J$, contrast C3 (J-I versus J-F) shows that, relative to the frozen-mid anchor, indicative anchoring worsens penalized shortfall for DDPG, TD3 and SAC, while the DQN effect is inconclusive. Finally, retraining with and without auction access asks a different question from the auction contribution of Section \ref{subsec:auction_contribution}, because the CLOB policy can adapt to the available closing liquidity. In this contrast (E-0 versus No-A, both without the explicit $H^\cl$ input), auction access lowers penalized shortfall for all four learners, with positive paired intervals.\vspace{0.5em}

Overall, in our experiments, policies that anticipate the call improve inventory-penalized shortfall relative to AS and TWAP, while their mean ordinary shortfalls remain higher than those of AS. In the hybrid historical-midprice simulations, the modeled auction contribution comes mainly from inventory-risk relief. Retraining with auction access also benefits all four learners in the synthetic setting. Dense auction credit improves learning despite the delay before settlement. The forecast is accurate relative to the contemporaneous midprice, but adding its explicit scalar input to a representation that already contains root-derived auction information worsens DQN's performance and has uncertain effects for the other learners. Using the forecast as a quote anchor can therefore worsen performance.\vspace{0.5em}

Our findings thus indicate that RL for liquidation and, ultimately, market making can extend to market structures beyond LOBs followed by closing auctions, such as workup sessions, the AHEAD mechanism \cite{duffie2017size,derchu2024ahead}, or sequences of periodic auctions \cite{milgrom1985economics,madhavan1992trading,budish2014implementation}.

\begin{figure}[H]

\centering
\includegraphics[width=0.75\linewidth]{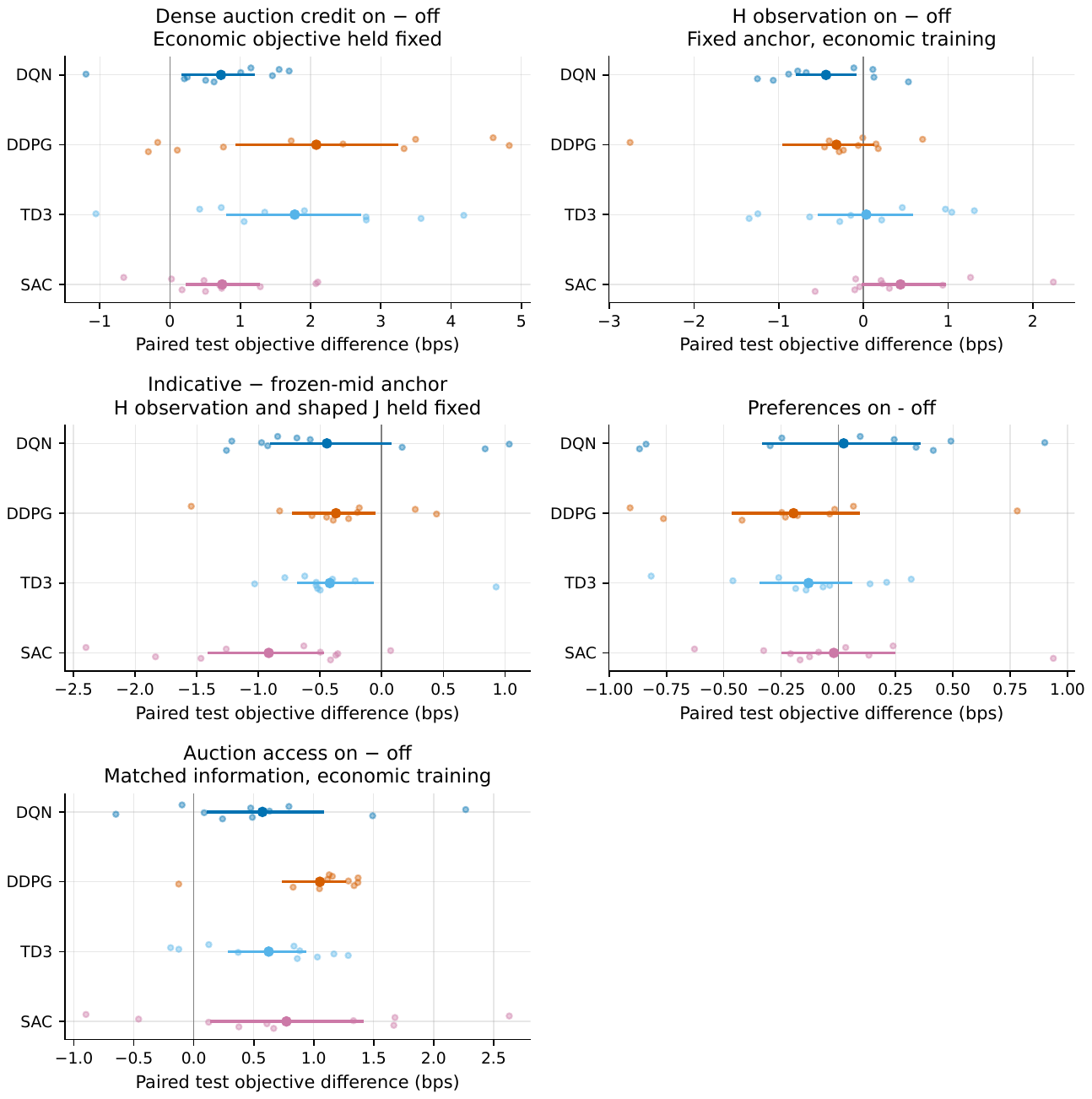}
\caption{Seed-level matched contrasts}
\label{fig:forecast_information}
\end{figure}

\begin{figure}[H]

\centering
\includegraphics[width=\linewidth]{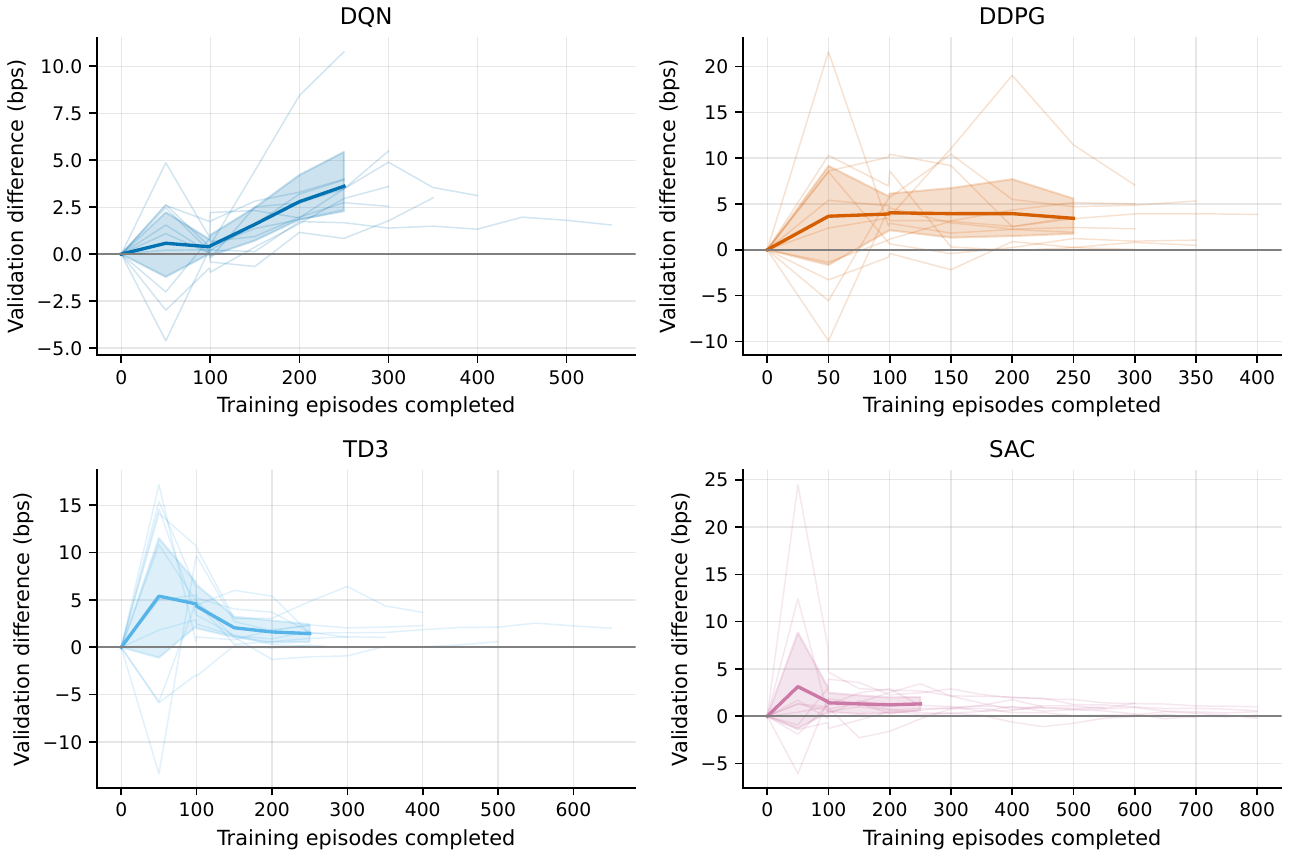}
\par\smallskip
\begin{subfigure}{0.49\linewidth}
\includegraphics[width=\linewidth]{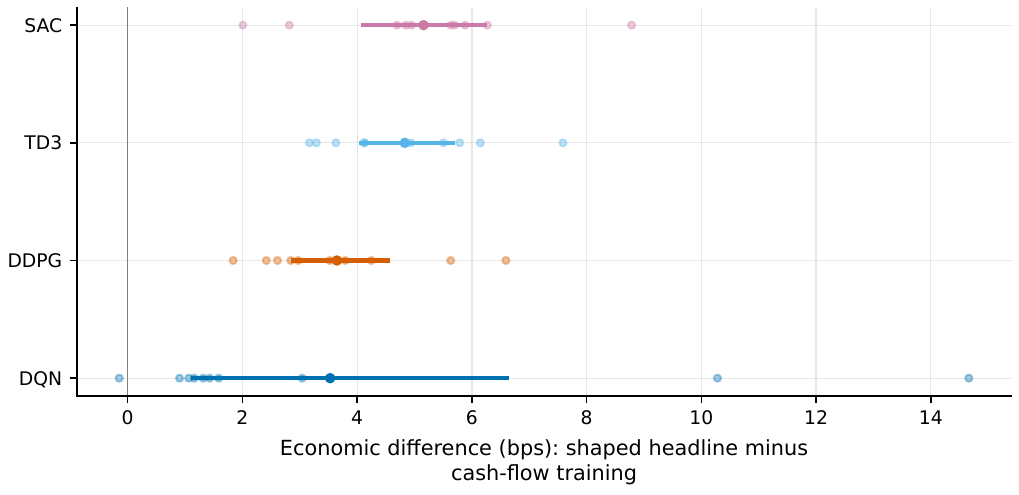}
\caption{C6: J-I$-$Cash-I}
\label{fig:learning_cash}
\end{subfigure}\hfill
\begin{subfigure}{0.49\linewidth}
\includegraphics[width=\linewidth]{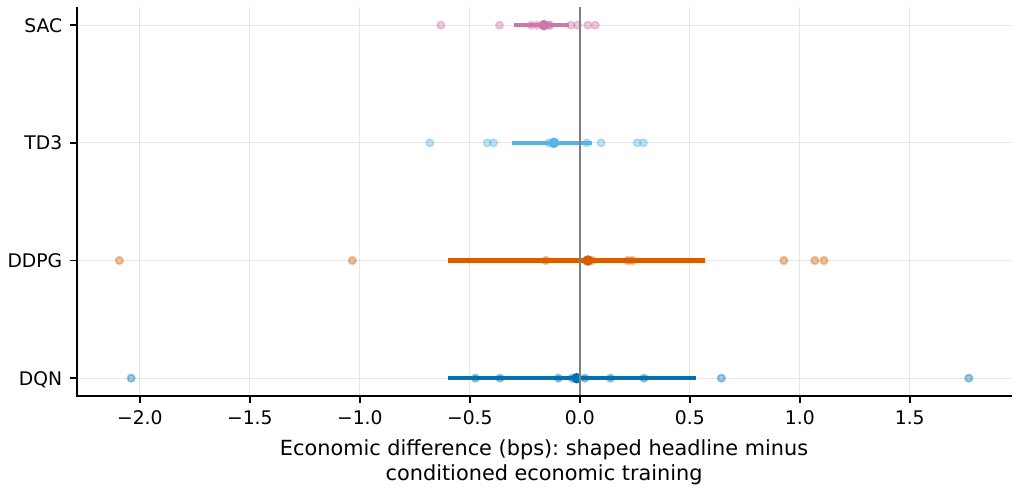}
\caption{C7: J-I$-$E-I}
\label{fig:learning_pref}
\end{subfigure}
\caption{Learning from delayed auction outcomes}
\label{fig:learning_guidance}
\end{figure}
\FloatBarrier

\section*{Code and Data Availability}

The source code, experiment configurations, saved numerical summaries underlying all tables and figures, and a \texttt{README.md} are available at
\begin{center}
\url{https://github.com/juliusgraf/learning-optimal-liquidation/tree/v0.1.0}.
\end{center}
The historical midprice paths were obtained through the Alpaca Market Data API and cannot be redistributed under its terms.

\newpage

\bibliographystyle{plain}
\bibliography{references}

\end{document}